%% file: main.tex
\documentclass[11pt]{article}
\input{head}

\def\Anonymity{0}

\title{New Algebrization Barriers to Circuit Lower Bounds \\via Communication Complexity of Missing-String}
\ifnum\Anonymity=0
\author{
    Lijie Chen\footnote{University of California, Berkeley. \texttt{\href{mailto:lijiechen@berkeley.edu}{lijiechen@berkeley.edu}}.}
    \and
    Yang Hu\footnote{Institute for Interdisciplinary Information Sciences, Tsinghua University. \texttt{\href{mailto:y-hu22@mails.tsinghua.edu.cn}{y-hu22@mails.tsinghua.edu.cn}}.}
    \and
    Hanlin Ren\footnote{Institute for Advanced Study. \texttt{\href{mailto:h4n1in.r3n@gmail.com}{h4n1in.r3n@gmail.com}}.}
}
\fi
\date{\today}

\begin{document}

\pagenumbering{gobble}

\maketitle

\begin{abstract}

    The \emph{algebrization barrier}, proposed by Aaronson and Wigderson (STOC '08, ToCT '09), captures the limitations of many complexity-theoretic techniques based on arithmetization. Notably, several circuit lower bounds that overcome the relativization barrier (Buhrman--Fortnow--Thierauf, CCC '98; Vinodchandran, TCS '05; Santhanam, STOC '07, SICOMP '09) remain subject to the algebrization barrier. 

    In this work, we establish several new algebrization barriers to circuit lower bounds by studying the communication complexity of the following problem, called $\textsc{XOR-Missing-String}$: For $m < 2^{n/2}$, Alice gets a list of $m$ strings $x_1, \dots, x_m\in\{0, 1\}^n$, Bob gets a list of $m$ strings $y_1, \dots, y_m\in\{0, 1\}^n$, and the goal is to output a string $s\in\{0, 1\}^n$ that is not equal to $x_i\oplus y_j$ for any $i, j\in [m]$. 
    \begin{enumerate}
        \item We construct an oracle $A_1$ and its multilinear extension $\widetilde{A_1}$ such that ${\sf PostBPE}^{\widetilde{A_1}}$ has linear-size $A_1$-oracle circuits on infinitely many input lengths. That is, proving ${\sf PostBPE}\not\subseteq \text{i.o.-}{\sf SIZE}[O(n)]$ requires non-algebrizing techniques. This barrier follows from a ${\sf PostBPP}$ communication lower bound for $\textsc{XOR-Missing-String}$. This is in contrast to the well-known algebrizing lower bound ${\sf MA_E}\, (\subseteq {\sf PostBPE}) \not\subseteq \P/_\poly$. %

        \item We construct an oracle $A_2$ and its multilinear extension $\widetilde{A_2}$ such that ${\sf BPE}^{\widetilde{A_2}}$ has linear-size $A_2$-oracle circuits on all input lengths. Previously, a similar barrier was demonstrated by Aaronson and Wigderson, but in their result, $\widetilde{A_2}$ is only a \emph{multiquadratic} extension of $A_2$. Our results show that communication complexity is more useful than previously thought for proving algebrization barriers, as Aaronson and Wigderson wrote that communication-based barriers were ``more contrived''. This serves as an example of how $\textsc{XOR-Missing-String}$ forms new connections between communication lower bounds and algebrization barriers.

        \item Finally, we study algebrization barriers to circuit lower bounds for $\sf MA_E$. Buhrman, Fortnow, and Thierauf proved a \emph{sub-half-exponential} circuit lower bound for $\sf MA_E$ via algebrizing techniques. Toward understanding whether the half-exponential bound can be improved, we define a natural subclass of $\sf MA_E$ that includes their hard $\sf MA_E$ language, and prove the following result: For every \emph{super-half-exponential} function $h(n)$, we construct an oracle $A_3$ and its multilinear extension $\widetilde{A_3}$ such that this natural subclass of ${\sf MA}_{\sf E}^{\widetilde{A_3}}$ has $h(n)$-size $A_3$-oracle circuits on all input lengths. This suggests that half-exponential might be the correct barrier for ${\sf MA_E}$ circuit lower bounds w.r.t.~algebrizing techniques. 
    \end{enumerate}
\end{abstract}

\newpage

\tableofcontents

\newpage

\newpage

\pagenumbering{arabic}

\input{intro}

\input{preliminaries}

\input{multilinear}

\input{missing_string_communication}

\input{bpp_ae}

\input{ma_half_exponential}

\appendix

\bibliographystyle{alphaurl}

{\small \bibliography{refs}}

\input{BFT_is_robust.tex}

\listoffixmes

\end{document}

%% file: head.tex
\usepackage{indentfirst}
\usepackage{amsmath}
\usepackage{amsthm}
\usepackage{bbm}
\usepackage{color}
\usepackage{eufrak}
\usepackage{verbatim}
\usepackage[makeroom]{cancel}
\usepackage{rotating}
\usepackage{float}
\usepackage[margin=1in]{geometry}
\usepackage{amssymb}
\usepackage{mathrsfs}
\usepackage{aliascnt}
\usepackage[pagebackref=true]{hyperref}
\usepackage{epsfig}
\usepackage{graphicx}
\usepackage{mleftright}
\usepackage[nottoc]{tocbibind}%
\usepackage[langfont=roman]{complexity}
\usepackage{indentfirst}
\usepackage{framed}
\usepackage[numbers]{natbib}
\usepackage{color}
\usepackage{graphicx}
\usepackage{mleftright}
\usepackage{float}
\usepackage{mdframed}
\usepackage{bm}
\usepackage{enumitem}
\usepackage{tikz}
\usepackage{multicol}
\usepackage{thmtools, thm-restate}
\usepackage{mathtools}
\usepackage{xspace}
\usetikzlibrary{positioning,shapes,shadows,arrows}

\usepackage[labelfont=bf]{caption}

\usepackage[linesnumbered,boxed,ruled,vlined,algo2e]{algorithm2e}

\usepackage[capitalize]{cleveref}

\definecolor{corlinks}{RGB}{100,0,100}
\definecolor{cormenu}{RGB}{100,0,100}
\definecolor{corurl}{RGB}{100,0,100}

\declaretheoremstyle[
spaceabove=\topsep, spacebelow=\topsep,
headfont=\normalfont\bfseries,
notefont=\bfseries, notebraces={}{},
bodyfont=\normalfont\itshape,
postheadspace=0.5em,
name={\ignorespaces},
numbered=no,
headpunct=.]
{mystyle}

\newtheorem{theorem}{Theorem}[section]

\newtheorem{problem}[theorem]{Problem}
\newtheorem{lemma}[theorem]{Lemma}
\newtheorem{corollary}[theorem]{Corollary}

\newtheorem{claim}[theorem]{Claim}

\theoremstyle{definition}
\newtheorem{definition}[theorem]{Definition}

\newenvironment{claimproof}
{\proof}
{\endproof}

\hypersetup{
	colorlinks=true,
	urlcolor=corlinks,
	linkcolor=corlinks,
	menucolor=cormenu,
	citecolor=corlinks,
	pdfborder={0 0 0}
}

\DeclareMathOperator{\io}{\mathrm{i.o.}\text{-}}

\newcommand{\MCSP}{\mathrm{MCSP}}

\newcommand{\calA}{\mathcal{A}}

\newcommand{\calC}{\mathcal{C}}

\newcommand{\calD}{\mathcal{D}}

\newcommand{\calI}{\mathcal{I}}

\newcommand{\calP}{\mathcal{P}}

\newcommand{\N}{\mathbb{N}}

\renewcommand{\epsilon}{\varepsilon}

\newcommand{\dt}{\mathrm{dt}}%

\renewcommand{\poly}{\mathrm{poly}}

\def\colorful{1}

\ifnum\colorful=1

\fi
\ifnum\colorful=0

\fi

\usepackage[inline,marginclue,index]{fixme}  %
\fxsetup{theme=color,mode=multiuser}

\definecolor{color1}{RGB}{46,134,193}
\FXRegisterAuthor{hanlin}{ahanlin}{\colorbox{color1}{\color{white}Hanlin}}
\FXRegisterAuthor{yang}{ayang}{\colorbox{orange}{\color{white}Yang}}

\FXRegisterAuthor{lijie}{alijie}{\colorbox{violet}{\color{white}Lijie}}

\FXRegisterAuthor{gpt}{agpt}{\colorbox{blue}{\color{white}GPT}}

\newcommand{\gpt}[1]{{\gptnote{#1}}}

\newcommand{\eps}{\varepsilon}

\newclass{\QuasiP}{QuasiP}
\newclass{\BPSUBEXP}{BPSUBEXP}
\newclass{\PostBPE}{PostBPE}
\newclass{\PostBPP}{PostBPP}
\newclass{\PostBPTIME}{PostBPTIME}
\newclass{\EXPH}{EXPH}
\newclass{\luCKT}{lu\text{-}CKT}
\newclass{\AME}{AM\text{-}E}
\newclass{\MATIME}{MATIME}
\newclass{\FMA}{FMA}
\newclass{\ZPEXP}{ZPEXP}
\newclass{\Socratic}{Socratic}

\newclass{\PTIME}{PTIME}
\newcommand{\MissingString}{\textnormal{\textsc{Missing-String}}\xspace}
\newcommand{\XorMissingString}{\textnormal{\textsc{XOR-Missing-String}}\xspace}

\newcommand{\pr}{\mathrm{pr}\text{-}}
\renewcommand{\MAE}{\MA_\E}
\newcommand{\RobMAE}{\mathrm{Rob}_{h(n)}\text{-}\MA_\E}

\usepackage{titlesec}

\newcommand{\defn}[1]{\emph{\boldmath\textbf{#1}}}

\usepackage{regexpatch}
\makeatletter
\xpatchcmd\thmt@restatable{%
\csname #2\@xa\endcsname\ifx\@nx#1\@nx\else[{#1}]\fi
}{%
\ifthmt@thisistheone
\csname #2\@xa\endcsname\ifx\@nx#1\@nx\else[{#1}]\fi
\else
\csname #2\@xa\endcsname[{Restated}]
\fi}{}{}
\makeatother

%% file: intro.tex
\section{Introduction}\label{sec: intro}

Proving unconditional circuit lower bounds is one of the major challenges in theoretical computer science, with the holy grail of proving $\NP\not\subseteq \P/_\poly$. \gpt{Notation: Prefer $\P/\poly$ (no underscore). Also consider introducing \EXPH and using it when referring to the exponential-time hierarchy.} Unfortunately, our progress toward this goal is barely satisfactory, as it is even open to prove a super-polynomial size lower bound for huge, exponential-time classes such as $\NEXP$. Only for even larger complexity classes such as $\Sigma_2\EXP$~\cite{DBLP:journals/iandc/Kannan82}, which are in (or beyond) the second level of exponential-time hierarchy, are super-polynomial lower bounds known. %

Our lack of progress in proving circuit lower bounds is partially explained by a series of \emph{barrier results} such as relativization~\cite{DBLP:journals/siamcomp/BakerGS75}, natural proofs~\cite{RazborovR97}, and algebrization~\cite{DBLP:journals/toct/AaronsonW09}. Among these barriers, relativization and algebrization are particularly relevant to lower bounds against \emph{unrestricted circuits} for \emph{large} complexity classes. For example, Wilson~\cite{DBLP:journals/jcss/Wilson85} constructed an oracle world where $\P^\NP$ has linear-size circuits, which explains our inability to prove fixed-polynomial size lower bounds for $\P^\NP$.

\paragraph{Missing-String: a duality-based approach to relativizing circuit lower bounds.} Previous relativization barriers for circuit lower bounds are proved in an \emph{ad hoc} fashion, which involves carefully analyzing the interaction between the oracle and the machines to be diagonalized~\cite{DBLP:journals/jcss/Wilson85, BuhrmanFT98,Aaronson06}. A recent paper by Vyas and Williams~\cite{VyasWilliams23} introduced the $\MissingString$ problem as a systematic approach to such relativization barriers:

\begin{problem}[$\MissingString(n, m)$]
    Let $n, m$ be such that $m < 2^n$. Given (query access to) a list of length-$n$ strings $x_1, x_2, \dots, x_m$, output a length-$n$ string that is not in this list.
\end{problem}

The query complexity of $\MissingString$ captures relativizing circuit lower bounds in the following sense: relativization barriers to proving $\calC\text{-}\EXP\not\subseteq\P/_\poly$ are essentially $\calC^\dt$ lower bounds for $\MissingString$. (Here, $\calC\text{-}\EXP$ is the exponential-time version of $\calC$ and $\calC^\dt$ means the decision-tree version of $\calC$.) This connection was made explicit in~\cite{VyasWilliams23}, who demonstrated an equivalence between relativization barriers to exponential lower bounds for $\Sigma_2\E$ and the non-existence of small depth-$3$ circuits for $\MissingString$.\footnote{$\E = \DTIME[2^{O(n)}]$ denotes \emph{single-exponential} time and $\EXP = \DTIME[2^{n^{O(1)}}]$ denotes \emph{exponential time}; classes such as $\Sigma_2\E$ and $\Sigma_2\EXP$ are defined analogously. Exponential time and single-exponential time are basically interchangeable in the context of super-polynomial lower bounds by a padding argument.}

Intriguingly, $\MissingString$ connects \emph{circuit lower bounds} with \emph{circuit upper bounds} and provides a \emph{duality}-based approach to both: Decision tree upper bounds for $\MissingString$ imply relativizing circuit lower bounds, and decision tree lower bounds for $\MissingString$ imply relativized worlds with circuit upper bounds (i.e., relativization barrier to circuit lower bounds). Besides providing a clean and systematic method for relativization barriers, this ``algorithmic'' perspective has indeed made progress in circuit lower bounds: By designing algorithms for the \emph{Range Avoidance} problem~\cite{KKMP21, DBLP:conf/focs/Korten21, DBLP:conf/focs/RenSW22}, which is the ``white-box'' version of $\MissingString$,\footnote{In the Range Avoidance problem, we are given the description of a circuit $C: \{0, 1\}^n \to \{0, 1\}^{n+1}$ and our goal is to output a string $y \in \{0, 1\}^{n+1}$ that is not in the range of $C$. It is easy to see that \emph{relativizing} algorithms for Range Avoidance are equivalent to decision trees for $\MissingString$.} recent work~\cite{CHR24, Li24} proved an exponential-size, relativizing lower bound for the complexity class $\Sigma_2\E$ (which also implies a quasi-polynomial size depth-$3$ circuit upper bound for $\MissingString$, settling the question in~\cite{VyasWilliams23}).

\paragraph{The quest of algebrization.}

However, the relativization barrier does not capture many circuit lower bounds that follow from \emph{nonrelativizing} results such as $\IP = \PSPACE$~\cite{LundFKN92, Shamir92} and $\MIP = \NEXP$~\cite{DBLP:journals/cc/BabaiFL91, DBLP:journals/cc/BabaiFNW93}, whose proofs are based on nonrelativizing proof techniques such as \emph{arithmetization}. This includes circuit lower bounds for $\PP$~\cite{DBLP:journals/tcs/Vinodchandran05} and $\MA$~\cite{BuhrmanFT98, DBLP:journals/siamcomp/Santhanam09}, both of which are \emph{provably nonrelativizing}~\cite{BuhrmanFT98, Aaronson06}.

To shed light on the limitations of such proof techniques, Aaronson and Wigderson~\cite{DBLP:journals/toct/AaronsonW09} proposed the \emph{algebrization} barrier: a lower bound statement $\calC \not\subseteq \calD$ \emph{algebrizes} if $\calC^{\widetilde A} \not\subseteq \calD^A$ for all oracles $A$ and all low-degree extensions $\widetilde A$ of~$A$. Aaronson and Wigderson showed that many arithmetization-based results indeed algebrize; in particular, $\MA_\E^{\widetilde A} \not\subseteq \P^A/_\poly$ for every oracle $A$ and low-degree extension $\widetilde A$ of $A$~\cite[Theorem 3.17]{DBLP:journals/toct/AaronsonW09}. In fact, all super-polynomial lower bounds against general circuits that we are aware of are algebrizing. \gpt{Claim strength: Consider softening (“to the best of our knowledge”) or adding a citation if available.} Thus, algebrization is a more suitable framework than relativization for capturing the limitations of our current techniques. %

Unfortunately, our understanding of the algebrization barrier remains primitive. For instance, several aspects of the algebrizing lower bound $\MA_\E \not\subseteq \P/_\poly$~\cite{BuhrmanFT98} remain poorly understood even with respect to algebrizing techniques:
\begin{itemize}
    \item This lower bound only holds \emph{infinitely-often}. That is, \cite{BuhrmanFT98} only exhibited a language $L \in \MA_\E$ where there are \emph{infinitely many} input lengths on which $L$ has high circuit complexity. Recently, the relativizing infinitely-often lower bound for $\Sigma_2\E$ was improved to almost everywhere by~\cite{Li24}, and it is tempting to ask whether the same improvement can be made with respect to this lower bound. Does $\MA_\E$ require large circuit complexity on \emph{every} input length? If so, can we show this via algebrizing proof techniques? %
    \item This lower bound is only \emph{sub-half-exponential}. Roughly speaking, a function $h$ is \emph{half-exponential} if $h(h(n)) \approx 2^n$. Half-exponential bounds appear naturally in many win-win analyses in complexity theory~\cite{MiltersenVW99}, and \cite{BuhrmanFT98} is no exception---it is only known how to prove a size-$h(n)$ lower bound for $\MA_\E$ when $h$ is smaller than half-exponential. The recent ``iterative win-win paradigm''~\cite{DBLP:conf/focs/CLORS23, CHR24, CLL25} provides new techniques for overcoming the half-exponential barrier, which has indeed improved the circuit lower bound for $\Sigma_2\EXP$ to a near-maximum ($2^n/n$) one~\cite{CHR24}. Can we use similar techniques to prove an exponential size lower bound for $\MA_\E$? If so, can we show this via algebrizing proof techniques?
\end{itemize}

Our lack of understanding on algebrization barriers raises the following question: 

\begin{quote}
	\centering
	Is there a \emph{duality-based approach} to algebrization barriers?
\end{quote}

In particular, can we establish algebrization barriers via lower bounds for $\MissingString$? It follows from~\cite{DBLP:journals/toct/AaronsonW09} that \emph{communication} lower bounds for $\MissingString$ imply algebrization barriers to circuit lower bounds.\footnote{In fact, lower bounds for $\MissingString$ in the \emph{algebraic query model} suffices. However, we find the algebraic query model somewhat counter-intuitive to work with (for example, the input to the query algorithms consists of a $\MissingString$ instance along with its \emph{low-degree extension}). The ``transfer principle''~\cite[Section 4.3]{DBLP:journals/toct/AaronsonW09} allows us to simulate such query algorithms by communication protocols, hence we choose to study communication complexity.} Hence, a na\"ive attempt is to prove communication lower bounds for $\MissingString$ and translate them into algebrization barriers. Unfortunately, it turns out that $\MissingString$ admits an efficient deterministic communication protocol by a simple binary search.\footnote{Suppose that $m \le 2^n/10$, Alice has inputs $x_1, x_2, \dots, x_m \in \{0, 1\}^n$ and Bob has inputs $y_1, y_2, \dots, y_m \in \{0, 1\}^n$. For each string $s \in \{0, 1\}^n$, it costs only $O(\log m)$ bits of communication to obtain the value $f(s) := |\{i: x_i \le s\}| + |\{i: y_i \le s\}|$. Hence, we can use binary search to find two (lexicographically) adjacent strings $\mathrm{pred}(s)$ and $s$ such that $f(\mathrm{pred}(s)) = f(s)$ by communicating $O(n\log m)$ bits. Clearly, $s$ is not in the set $\{x_1, \dots, x_m, y_1, \dots, y_m\}$.} %

The main conceptual contribution of this paper is the following communication problem:

\begin{problem}[{$\XorMissingString(n, m)$}]
    Let $n, m$ be such that $m < 2^{n/2}$. Alice receives $m$ strings $x_1, x_2, \dots, x_m \in \{0, 1\}^n$ and Bob receives $y_1, y_2, \dots, y_m \in \{0, 1\}^n$. Their goal is to communicate with each other and output an $n$-bit string $s$ such that $s$ is not equal to $x_i\oplus y_j$ for every $1\le i, j\le m$. Here $\oplus$ denotes the bit-wise XOR operation over strings.
\end{problem}

We show that $\XorMissingString$ is hard for various communication complexity models, and transfer our communication complexity lower bounds to algebrization barriers. We now discuss our results in more detail.

\subsection{Barriers to \texorpdfstring{$\pr\PostBPE$}{pr-PostBPE} Circuit Lower Bounds}

Our first result is an algebrization barrier to proving $\pr\PostBPE \not\subseteq \io\SIZE[O(n)]$, i.e., an almost-everywhere circuit lower bound for $\pr\PostBPE$. Here, $\pr\PostBPE$ is the class of promise problems decided by a probabilistic exponential-time machine with \emph{postselection}~\cite{DBLP:journals/siamcomp/HanHT97}: conditioned on some event (which may happen with exponentially small probability), the machine outputs the correct answer with high probability. Postselection is a powerful computational resource: $\PostBPP$ contains $\MA$ (thus $\NP$)~\cite{DBLP:journals/siamcomp/HanHT97} and $\PostBQP$ (polynomial-time \emph{quantum} computation with postselection) equals $\PP$~\cite{aaronson2005quantum}.\footnote{For example, a $\PostBPP$ machine can solve $\NP$-complete problems as follows. With exponentially small probability, output ``No'' and halt. If this does not happen, choose an $\NP$ witness uniformly at random and ``kill yourself'' if this witness is invalid (that is, we \emph{postselect} on the event that our witness is valid). If the algorithm survives, it outputs ``Yes.'' Conditioned on survival, with high probability, our algorithm outputs ``Yes'' on Yes instances and outputs ``No'' on No instances.} \gpt{Style: Consider avoiding colloquial “kill yourself”; e.g., “abort” or “reject if the witness is invalid.”}
{
\def\ImplyMA{1}
\def\Oracle{A_1}
\begin{restatable}{theorem}{ThmioPostBPPBarrier}\label{thm:io_postbpp_barrier}
    There exists an oracle $\Oracle$ and its multilinear extension $\widetilde{\Oracle}$ such that
    \[
    \pr\PostBPE^{\widetilde{\Oracle}} \subseteq \io\SIZE^{\Oracle}[O(n)].
    \]
    \ifnum\ImplyMA=1
    In particular, this also implies
    \[
    \pr\MA_\E^{\widetilde{\Oracle}} \subseteq \io\SIZE^{\Oracle}[O(n)].
    \]
    \fi
\end{restatable}
}
The proof of $\pr\MA_\E\not\subseteq\P/_\poly$ is algebrizing~\cite{BuhrmanFT98,DBLP:journals/siamcomp/Santhanam09,DBLP:journals/toct/AaronsonW09}. Our result implies that improving this circuit lower bound to almost-everywhere requires non-algebrizing techniques.

\autoref{thm:io_postbpp_barrier} follows from a $\PostBPP$ communication lower bound for $\XorMissingString$:

\begin{restatable}{theorem}{LemmaXORMSLB}
    \label{lem:xor_ms_lb}
    Let $n\ge 1$ and $20n\le m<2^{n/2}$ be integers. Any $\PostBPP$ communication protocol that solves $\XorMissingString(n,m)$ with error $\le 2^{-5n}$ must have communication complexity $\Omega(m)$.
\end{restatable}

Assume that $n \ll m \ll 2^{o(n)}$. A trivial protocol for $\XorMissingString$ is to output a uniformly random $n$-bit string, which has communication complexity $O(n)$ and error $m^2/2^n$. Thus, \autoref{lem:xor_ms_lb} states that even if postselection is allowed, if we want to beat the error bound of this trivial protocol ($\approx 2^{-n}$), then the amount of communication needed is close to the maximum ($\Omega(m)$).

Moreover, since the error of any \emph{pseudodeterministic}\footnote{A randomized protocol for a search problem is \emph{pseudodeterministic}~\cite{DBLP:journals/eccc/GatG11} if there is a \emph{canonical} output that is correct and outputted with probability $\ge 2/3$. The trivial protocol for $\XorMissingString$ that outputs a random guess is not pseudodeterministic, since running it two times using different randomness is likely to yield different answers.} communication protocol can be reduced by a $2^{-k}$ factor by repeating the protocol $\Theta(k)$ times and taking the majority answer, \autoref{lem:xor_ms_lb} implies that any pseudodeterministic $\PostBPP$ protocol that solves $\XorMissingString(n, m)$ (with correct probability, say, $2/3$) must have communication complexity $\Omega(m/n)$. This stands in stark contrast to the case without pseudodeterministic constraints, as the trivial protocol is correct with probability $1 - m^2/2^n \gg 2/3$. We also remark that lower bounds against pseudodeterministic $\PostBPP$ communication protocols suffice for constructing an oracle $A$ such that $\PostBPE^{\widetilde A}$ has small $A$-oracle circuits; however, we will need the full power of \autoref{lem:xor_ms_lb} against every low-error protocol to prove that \emph{the promise version} of $\PostBPE^{\widetilde A}$ has small $A$-oracle circuits.

\subsection{Barriers to \texorpdfstring{$\BPE$}{BPE} Circuit Lower Bounds}
Our second result is an algebrization barrier to proving $\BPE \not\subseteq \SIZE[O(n)]$, i.e., an \emph{infinitely-often} circuit lower bound for $\BPE$:
{
\def\Oracle{A_2}
\begin{restatable}{theorem}{ThmBPEBarrier}\label{thm: intro BPE}
    There exists an oracle $\Oracle$ and its multilinear extension $\widetilde{\Oracle}$ such that
    \[\BPE^{\widetilde{\Oracle}} \subseteq \SIZE^{\Oracle}[O(n)].\]
\end{restatable}
 \gpt{Label hygiene: The label key contains spaces (\texttt{thm: intro BPE}). Consider \texttt{thm:intro-BPE} (and similarly elsewhere) to avoid fragile references.}
}

Previously, Aaronson and Wigderson~\cite{DBLP:journals/toct/AaronsonW09} constructed an oracle $A$ and its \emph{multiquadratic} extension $\widetilde A$ such that ${\sf BPEXP}^{\widetilde A} \subseteq \P^A/_\poly$. Their proof works with the algebraic query model directly, and it is unclear how to prove the same result for multilinear extensions using their techniques. %

Besides demonstrating the power of our communication- and duality-based approach, an algebrization barrier w.r.t.~multilinear extension also aligns better with the notion of \emph{affine relativization}~\cite{DBLP:journals/toct/AydinliogluB18}. A statement $\calC \subseteq \calD$ is said to \emph{affine relativize} if $\calC^{\widetilde A} \subseteq \calD^{\widetilde A}$ for every $\widetilde A$ that is the multilinear extension of some oracle $A$. The crucial difference from algebrization (in the sense of \cite{DBLP:journals/toct/AaronsonW09}) is that the left-hand side is also required to relativize with the multilinear extension. Affine relativization is arguably a cleaner and more robust notion than algebrization.\footnote{For example, algebrization is not necessarily closed under \emph{modus ponens}: Although it requires non-algebrizing techniques to prove, say, $\NEXP\not\subseteq \P/_\poly$, it might still be possible to exhibit a complexity class $\calC$ and prove both $\calC \subseteq \NEXP$ and $\calC \not\subseteq \P/_\poly$ via algebrizing techniques. In contrast, affine relativization is clearly closed under \emph{modus ponens}. \label{footnote: modus ponens}} It is important to only take multilinear extensions, as \cite{DBLP:journals/toct/AydinliogluB18} was only able to show that, e.g., $\PSPACE^{\widetilde A} \subseteq \IP^{\widetilde A}$ for every multilinear oracle $\widetilde A$; it is unclear if such a statement holds when $\widetilde A$ is merely multiquadratic. (This is also why Aaronson and Wigderson~\cite{DBLP:journals/toct/AaronsonW09} choose to relativize $\IP$ with the low-degree extension of $A$ but to relativize $\PSPACE$ with $A$ itself.) In this regard, \autoref{thm: intro BPE} also shows that $\BPE\not\subseteq\P/_\poly$ is not affine relativizing (since $\BPE^{\widetilde A} \subseteq \SIZE^A[O(n)]\subseteq \SIZE^{\widetilde A}[O(n)]$).%

To obtain \autoref{thm: intro BPE}, we prove a lower bound for $\XorMissingString$ \emph{against multiple pseudodeterministic $\BPP$ protocols simultaneously} in the following model. \begin{itemize}
    \item There are $t$ communication protocols $Q_1, Q_2, \dots, Q_t$ and $t$ inputs $(X_1, Y_1)$, $(X_2, Y_2)$, $\dots$, $(X_t, Y_t)$.
    
    \item The goal of protocol $Q_i$ is to solve the input $(X_i, Y_i)$. However, each protocol is given the inputs to other protocols as well. More precisely, in each $Q_i$, Alice gets $(X_1, X_2, \dots, X_t)$ as inputs and Bob gets $(Y_1, Y_2, \dots, Y_t)$ as inputs.\footnote{We remark that in the case of $\XorMissingString$, each $X_i$ and $Y_i$ is already a list of strings, which means Alice and Bob get $t$ lists of strings each.}
    \item We say that the protocols \defn{succeed} if \emph{at least one} of the protocols $Q_i$ outputs a correct answer.
    \item For proving algebrization barriers, each $Q_i$ and $(X_i, Y_i)$ should have a different size and correspond to different input lengths of the oracle, but we omit this detail in the informal description here. We refer the reader to~\autoref{lemma: finite barrier} for details.
\end{itemize}

The point of this model is that the protocols are allowed to look at other protocols' inputs and to perform \emph{win-win analyses}: the correctness of $Q_i$ may rely on the \emph{failure} of some other protocol $Q_{i'}$.\footnote{
    If the possibility of such win-win analyses sounds surprising, it may help to compare it with the following classic puzzle. There are $N$ prisoners, each assigned a (not necessarily distinct) number between $0$ and $N-1$. Each prisoner can see others' numbers but not their own. Each prisoner must guess their own number simultaneously. If at least one prisoner guesses correctly, they are all set free; if none of them do, they are all executed.

    There is a simple solution that guarantees the prisoners' freedom. The $i$-th prisoner ($0\le i\le N-1$) assumes that the total sum of all numbers is congruent to $i$ modulo $N$, and then infers their own number as
    \[(i-\text{sum of other prisoners' numbers})\bmod N.\]
    Exactly one assumption will be correct, so that prisoner will guess their own number correctly.
}
Indeed, this scenario models a variety of win-win analyses in complexity theory; see e.g., \cite[Section 1.4.2]{CHR24} and \cite[Section 5.2]{LuORS24}. Circuit lower bounds for $\MA$~\cite{BuhrmanFT98,DBLP:journals/siamcomp/Santhanam09} can also be modeled as win-win algebraic query algorithms for $\MissingString$; see~\autoref{sec: BFT as algebraic query algorithms}.

We show that efficient pseudodeterministic protocols require large communication complexity to solve $\XorMissingString$, even if they receive multiple instances and are allowed to perform win-win analyses in the above sense.

\begin{theorem}[Informal and simplified version of \autoref{lemma: finite barrier}]\label{thm: BPP cc LB informal}
    In the above model, suppose that each $(X_i, Y_i)$ is a $\XorMissingString(n_i, m_i)$ instance but the communication complexity of each $Q_i$ is ``much less'' than $m_i$. Then there exists a sequence of inputs $\{(X_i, Y_i)\}_{i\in [t]}$ such that every $Q_i$ fails to solve its corresponding instance $(X_i, Y_i)$ pseudodeterministically.
\end{theorem}

\subsection{Barriers to \texorpdfstring{$\MA_\E$}{MAE} Circuit Lower Bounds}

Finally, we present algebrization barriers to improving the half-exponential lower bound for $\MA_\E$~\cite{BuhrmanFT98}. While we are unable to fully resolve the algebrizing circuit complexity of $\MA_\E$, we show that for a natural subclass of $\MA_\E$ which includes the hard language in~\cite{BuhrmanFT98}, non-algebrizing techniques are required to go beyond half-exponential bounds.

Recall that a standard algorithm $M^{\widetilde A}$ in $\MAE^{\widetilde A}$ defines a hard language if, \emph{relative to this particular oracle $\widetilde A$}, $M^{\widetilde A}$ satisfies the $\MA$ promise and defines a language without small $A$-oracle circuits; we do not care about the behavior of $M^{\widetilde B}$ for other oracles $\widetilde B$. In contrast, we say that an $\MA_\E$ machine $M$ is a \defn{robust} machine for defining a hard language (or simply ``is robust''), if for \emph{any} oracle $B$ and its multilinear extension $\widetilde B$ such that $M^{\widetilde B}$ satisfies the $\MA$ promise, the language computed by $M^{\widetilde B}$ does not have small $B$-oracle circuits. The use of interactive proofs in~\cite{BuhrmanFT98} naturally leads to hard languages defined by robust $\MA$ machines; see \autoref{app:bft} for details.\footnote{In fact, the definition of robust $\MA$ resembles closer to the class $\MA\cap\coMA$ compared to $\MA$ itself. This is natural since \emph{single-valued $\FMA$-constructions of hard truth tables} (see, e.g., \cite[Section 1.2.1]{CHR24}) give rise to circuit lower bounds for $\MA\cap\coMA$, and indeed the lower bound proved in~\cite{BuhrmanFT98} is for hard languages in $\MA_\E\cap\coMA_\E$. For ease of notation we use ``robust $\MA$'' ($\text{Rob-}\MA$) instead of ``robust $\MA\cap\coMA$'' ($\text{Rob-}\MA\cap\coMA$).} %

We use $\text{Rob-}\MA^A$ (resp.~$\text{Rob-}\MAE^A$) to denote the class of languages computed by robust $\MA$ (resp.~$\MAE$) algorithms with access to oracle $A$. Our main result is that proving a super-half-exponential bound for languages in $\E$ or robust $\MA_\E$ would require non-algebrizing techniques:

\begin{theorem}[{Informal version of \autoref{thm:ma_ae_barrier}}]\label{thm: intro MAE}
    There exists an oracle $A_3$ and its multilinear extension $\widetilde{A_3}$ such that both $\E^{\widetilde{A_3}}$ and $\mathrm{Rob}\text{-}\MA_\E^{\widetilde{A_3}}$ admit half-exponential size $A_3$-oracle circuits.
\end{theorem}

\paragraph{Why study robust $\MAE$?} We present two reasons for studying the notion of robust $\MAE$.

Firstly, \autoref{thm: intro MAE} is tight in the sense that, in the two cases of the win-win argument in~\cite{BuhrmanFT98}, the hard language is either in $\E$ (i.e., deterministic exponential time), or in robust $\MA_\E$. Therefore, a slight adaptation of~\cite{BuhrmanFT98} proves that for every oracle $A$ and its multilinear extension $\widetilde A$, $\E^{\widetilde A} \cup \mathrm{Rob}\text{-}\MA_\E^{\widetilde A}$ does not have sub-half-exponential size $A$-oracle circuits. This is shown in \autoref{app:bft}.

Secondly, while the classes $\E$ and $\RobMAE$ may initially appear arbitrary, they in fact capture two fundamental types of failures for algorithms in $\MA_\E \cap \co\MA_\E$:

\begin{itemize}
\item For $M \in \RobMAE$, if $M^{\widetilde{A}}$ fails to achieve high circuit complexity, it must be due to an input $x$ on which $M^{\widetilde{A}}$ is semantically incorrect. That is, the corresponding truth table entry is undefined. We call this a failure due to \textbf{no positive}.

\item In contrast, for \( M \in \E \), the value \( M^{\widetilde{A}}(x) \) is always well-defined for all \( x \), regardless of the oracle \( A \). If \( M^{\widetilde{A}} \) has low circuit complexity, it must be because its truth table is easy for \( A \)-oracle circuits. We call this a failure due to a \textbf{false positive}.

\end{itemize}

In general, for an algorithm $M \in \MA_\E \cap \co\MA_\E$, failure to achieve high circuit complexity can be attributed to both types, depending on the oracle. Thus, $\E$ and $\RobMAE$ represent two extreme cases within $\MA_\E \cap \co\MA_\E$.

Our proof of \autoref{thm:ma_ae_barrier} entails techniques that can force the two types of failures to occur. In this sense, the proof says something fundamental about $\MAE\cap\co\MAE$. However, our current methods are limited to handling algorithms that exhibit only a single failure mode. For more general algorithms in $\MA_\E \cap \co\MA_\E$ that can have mixed types of failures, our understanding remains incomplete.

\paragraph{Proof of \autoref{thm: intro MAE} via a communication lower bound.} To prove algebrization barriers to circuit lower bounds for robust $\MA_\E$, we need to understand the \emph{robust} $\MA$ communication complexity of $\XorMissingString$. Here, an $\MA$ protocol $P$ is called \emph{robust} if on every input $(X, Y)$, Merlin (i.e., the prover) cannot convince the protocol to output a non-solution for $(X, Y)$ except with very small probability. Underlying \autoref{thm: intro MAE} is the following lower bound stating that efficient $\MA$ protocols for $\XorMissingString$ that are robust can only be correct on a tiny fraction of inputs:%

\begin{restatable}{lemma}{ZeroErrorLB}
    \label{lem:socratic_lb}
    Let $n\ge 1,1\le m<2^{n/2}$ be parameters. Let $P$ be a robust $\FMA$ communication protocol of complexity $C$ attempting to solve $\XorMissingString(n,m)$. Then the fraction of inputs that $P$ solves is at most
    \begin{align*}
        2^{-\Omega(m/C)+O(C+n)}.
    \end{align*}
\end{restatable}

Intuitively, since there is no efficient way to verify whether a string is a solution, any protocol that solves a large number of instances of $\XorMissingString$ is essentially guessing the answer (similar to the na\"ive protocol that succeeds with probability $1-2^{-O(n)}$), which means that it must make mistakes. Since robust protocols are not allowed to output false positives, they can only be correct on a tiny fraction of inputs.

Although \autoref{lem:socratic_lb} is useful for analyzing the behaviors of robust $\MA$ protocols, it does not directly imply any \emph{half-exponential} bounds. Indeed, we need to carefully diagonalize against both $\E^{\widetilde A}$ and $\RobMAE^{\widetilde A}$ simultaneously, which leads to a half-exponential bound. We discuss this in more detail in \autoref{sec: tech overview}.%

\subsection{Technical Overview}\label{sec: tech overview}

Next, we briefly overview our proof techniques. In fact, the engine behind all of our communication lower bounds is the following observation about $\XorMissingString$:%
\begin{lemma}[{Informal version of~\autoref{lem:rect_error_prob_lb}}]
    \label{lem:rect_error_prob_lb_informal}
    For every large enough rectangle $R = \mathcal{X} \times \mathcal{Y}$ and every answer $s$, there is a large enough subrectangle $R' = \mathcal{X}' \times \mathcal{Y}'$ ($\mathcal{X}' \subseteq \mathcal{X}$ and $\mathcal{Y}'\subseteq \mathcal{Y}$) such that $s$ is a \emph{wrong answer} for the $\XorMissingString$ problem on every input $(X, Y) \in R'$.
\end{lemma}

With this lemma in hand, it is not hard to prove \autoref{lem:xor_ms_lb}. We characterize a $\PostBPP$ communication protocol as a distribution over labeled rectangles (using \emph{approximate majority covers}~\cite{Klauck2003}), where for each input $(X,Y)$, the output of the protocol is given by the label of a randomly selected rectangle that contains $(X,Y)$. Using \autoref{lem:rect_error_prob_lb_informal}, for every large enough labeled rectangle, there exists a large subrectangle in which the label is never a solution to $\XorMissingString$. \autoref{lem:rect_error_prob_lb_informal} thus translates to a lower bound for the success probability of large labeled rectangles (i.e., the probability that the label is a solution, over a random input in the rectangle). This then translates to a lower bound for (weighted) average success probability over a random input, because the distribution consists mostly of large rectangles when the protocol is efficient. This lower bound implies that efficient protocols cannot have very high success probability.

However, more work needs to be done to obtain our other communication lower bounds and algebrization barriers.

\paragraph{BPP communication lower bounds.} We prove our pseudodeterministic $\BPP$ communication lower bound in the ``win-win model'' (\autoref{thm: BPP cc LB informal}) via a reduction to the $\PostBPP$ communication lower bound (\autoref{lem:xor_ms_lb}). We first use an induction on the number of protocols; assume towards a contradiction that \autoref{thm: BPP cc LB informal} holds for any sequence of $t-1$ protocols but does not hold for some sequence of $t$ protocols $Q_1, \dots, Q_t$. Then the following $\PostBPP$ communication protocol solves $\XorMissingString$ efficiently: Given an input $(X, Y)$, guess $t-1$ inputs $\{(X_i, Y_i)\}_{i\in [t-1]}$ uniformly at random, set $(X_t, Y_t) := (X, Y)$, and \emph{postselect} on the event that all of the first $t-1$ protocols fail. That is, for every $1\le i \le t-1$, $Q_i$ fails to solve $(X_i, Y_i)$ given $\{(X_i, Y_i)\}_{i\in [t]}$. By our induction hypothesis, this event happens with probability strictly greater than $0$. If this happens, then $Q_t$ solves $(X_t, Y_t)$ correctly. This contradicts \autoref{lem:xor_ms_lb}.

For clarity, we have made a crucial simplification in the above description: To obtain algebrization barriers, we also need to prove lower bounds against \emph{list-solvers} for $\XorMissingString$, which are communication protocols that output a \emph{small set} of answers such that one of the answers is correct. Fortunately, it is not hard to adapt the proof of \autoref{lem:xor_ms_lb} to prove a $\PostBPP$ communication lower bound for list solvers; see \autoref{sec: list-solvers}.

\paragraph{$\MA$ communication lower bounds against robust protocols.} %
A robust protocol can be described as a collection of randomized verifiers $V_{w,\pi}$, where $V_{w,\pi}(X,Y)=1$ means that when the input is $(X,Y)$ and $\pi$ is given as proof, the protocol accepts $w$ as output. In the formal proof, we apply error reduction on the verifiers, so that when $w$ is not a solution to $(X,Y)$, $V_{w,\pi}(X,Y)=1$ with extremely small probability. 

Thus, if a protocol has large success probability, there must exist some answer string $w$ and some proof $\pi$, such that over a random input $(X,Y)$, $\Pr[V_{w,\pi}(X,Y)=1]$ is large. We interpret this $V_{w,\pi}$ as a distribution over labeled rectangles, and apply \autoref{lem:rect_error_prob_lb_informal} to show that conditioned on $V_{w,\pi}(X,Y)=1$, the verifier makes many mistakes (i.e., $\Pr[V_{w,\pi}(X,Y)=1\land w\text{ is not a solution to }(X,Y)]$ is large). This implies that the protocol must make mistakes on some $(X,Y)$, so it cannot be robust.

\paragraph{The half-exponential barrier for robust $\MA_\E$.} We formulate the problem in terms of refuting communication protocols, where we need to refute a finite set of protocols (deterministic or robust) on each input length. More formally, to construct an algebrized world with circuit upper bound $h(n)$, our oracle encodes one instance $(X_i, Y_i)$ of $\XorMissingString(2^i, 2^{h(i)})$ for each input length $i$. Each protocol being diagonalized is dedicated to solving the one instance that corresponds to its input length (it can nevertheless see the other instances).

We employ an inductive approach to construct the oracle. At step $i$, we maintain a rectangle\footnote{We interpret our oracles as the concatenation of Alice's and Bob's inputs, so it is reasonable to talk about a ``rectangle'' of oracles.} $R_i$ of oracles, such that the rectangle is relatively large, and any protocol that works on input length at most $i$ fails to solve its instance on any oracle $A\in R_i$.

The goal of every step is then to slightly shrink the rectangle $R_{i-1}$ into $R_i$, so that the protocols working on length $i$ are refuted. To achieve this, we employ a combination of two different strategies:

\begin{itemize}
    \item \textbf{Refuting deterministic protocols:} Given a deterministic protocol $P$ working on length $i$ and a large enough rectangle $R_i$, we can always find a large subrectangle $R'\subseteq R_i$ such that $P$ outputs the same answer on every oracle in $R'$ (this is by definition of deterministic protocols). By \autoref{lem:rect_error_prob_lb_informal}, there exists a large subrectangle $R''\subseteq R'$, such that for any oracle $A\in R''$, $P$ does not solve $(X_i,Y_i)$ on $A$. We then update $R_i\gets R''$ to refute $P$.
    \item \textbf{Refuting robust protocols:} Given a robust protocol $Q$ working on length $i$ and a large enough rectangle $R_i$, a random oracle from $R_i$ refutes $Q$ with high probability. This is true as long as the density of $R_i$ is sufficiently larger than the success probability of $Q$ on a uniformly random input (as upper bounded in \autoref{lem:socratic_lb}).
\end{itemize}
On their own, both strategies work extremely well: The first strategy is very similar to the proof of $\NEXP^{\widetilde{A}}\subset \P^A/\poly$ by \cite{DBLP:journals/toct/AaronsonW09}, and can be used to obtain much better circuit upper bounds than half-exponential. For the second strategy, since robust protocols are very weak, a randomly selected oracle would refute all robust protocols with high probability.

However, it is unclear how to combine the two strategies. On the one hand, after refuting a deterministic protocol $P$ by the first strategy, the resulting rectangle $R_i$ may be too small for the second strategy to work. On the other hand, although the set of oracles $R'\subseteq R_i$ that refutes a robust protocol $Q$ is large, it may not be a rectangle, so we cannot use it to refute the next deterministic protocol.

To integrate these two strategies, we refute each robust protocol at a carefully chosen time step. Suppose that we refute a robust protocol $Q_i$ working on length $i$ \emph{after} all deterministic protocols working on lengths $< k$ are refuted, where $k = k(i)$ is some function on $i$. Also, recall that we want to prove an algebrized circuit upper bound of size $h(n)$. %
\begin{itemize}
    \item After refuting the deterministic protocols working on lengths $< k$, we end up with a rectangle that contains roughly a $2^{-2^{O(k)}}$ fraction of all oracles. This is because the deterministic protocols working on lengths $< k$ run in time $2^{O(k)}$. In contrast, \autoref{lem:socratic_lb} implies that $Q_i$ only solves a $\approx 2^{-2^{h(i)}}$ fraction of oracles (since it attempts to solve an instance of $\XorMissingString(2^i,2^{h(i)})$). Therefore, if $2^{-2^{O(k)}}\gg 2^{-2^{h(i)}}$, then we can refute $Q_i$.
    \item On the other hand, after refuting $Q_i$, we still need to refute the deterministic protocols working on length $k$. Since these protocols attempt to solve $\XorMissingString(2^k, 2^{h(k)})$, to apply \autoref{lem:rect_error_prob_lb_informal}, this requires our current rectangle $R_k$ to have size at least roughly $2^{-2^{h(k)}}$. Since $Q_i$ is equivalent to a deterministic protocol running in time $2^{2^{O(i)}}$ (by enumerating the random bits and Merlin's proof), after refuting it, we still have a rectangle of density $\approx 2^{-2^{2^{O(i)}}}$. We thus need that $2^{-2^{2^{O(i)}}}\gg2^{-2^{h(k)}}$.
\end{itemize}
Overall, we require that $k\ll h(i)$ and $2^i\ll h(k)$. This implies that $h$ must be super-half-exponential.%

\subsection{Related Works}

\paragraph{Prior work on super-polynomial circuit lower bounds.} Kannan's seminal work~\cite{DBLP:journals/iandc/Kannan82} proved a super-polynomial size lower bound for $\Sigma_2\EXP$; since then, a sequence of work has proved circuit lower bounds for various complexity classes such as $\ZPEXP^\NP$~\cite{KoblerW98,BshoutyCGKT96}, $\S_2\EXP$~\cite{CaiCHO05, Cai01a}, $\PEXP$~\cite{DBLP:journals/tcs/Vinodchandran05,Aaronson06}, $\MA_\EXP$~\cite{BuhrmanFT98, DBLP:journals/siamcomp/Santhanam09}, $\ZPEXP^\MCSP$~\cite{DBLP:conf/coco/ImpagliazzoKV18,HiraharaLR23}, and more \cite{StockmeyerM02, CLL25}. Such lower bounds are usually proved by \emph{Karp--Lipton} theorems~\cite{DBLP:conf/stoc/KarpL80, DBLP:journals/cc/BabaiFNW93, DBLP:conf/coco/ChenMMW19}: If a large uniform class (usually~$\PH$ or $\EXP$) admits polynomial-size circuits, then it collapses to a smaller uniform class (such as $\Sigma_2\P$). As explained in~\cite{MiltersenVW99} and \cite[Section 1.4.1]{CHR24}, such a strategy naturally yields a \emph{half-exponential} bound. Recently, \cite{CHR24, Li24} proved near-maximum ($2^n/n$) circuit lower bounds for the classes $\Sigma_2\EXP$, $\ZPEXP^\NP$, and $\S_2\EXP$, improving upon previous half-exponential bounds. Near-maximum circuit lower bounds are also known for several classes with \emph{subexponential} amount of advice bits, such as $\mathsf{BPEXP}^\MCSP/_{2^{n^\eps}}$~\cite{HiraharaLR23} and $\mathsf{AMEXP}/_{2^{n^\eps}}$~\cite{CLL25}.%

\paragraph{Algebrization.} The interactive proof results such as $\IP = \PSPACE$~\cite{LundFKN92, Shamir92} and $\MIP = \NEXP$~\cite{DBLP:journals/cc/BabaiFL91} generated much excitement among complexity theorists, as they are the first ``truly compelling'' (\cite{DBLP:conf/sigal/Allender90}) non-relativizing results in complexity theory. There has been much discussion on the extent to which these results are non-relativizing, and how to adapt the relativization barrier to accommodate these results~\cite{arora1992relativizing, DBLP:journals/eatcs/Fortnow94}. Arguably, the most influential work in this direction is that of Aaronson and Wigderson~\cite{DBLP:journals/toct/AaronsonW09} on the algebrization barrier, which nicely captures interactive proof results and arithmetization-based techniques. However, Aaronson and Wigderson's algebrization barrier has its own subtleties (such as not being closed under \emph{modus ponens}, see \autoref{footnote: modus ponens}), which leads to several proposed refinements of its definition~\cite{DBLP:conf/stoc/ImpagliazzoKK09, DBLP:journals/toct/AydinliogluB18}. We also mention the \emph{bounded relativization} barrier, recently proposed by Hirahara, Lu, and Ren~\cite{HiraharaLR23}, which attempts to capture the interactive proof results entirely within the original Baker--Gill--Solovay~\cite{DBLP:journals/siamcomp/BakerGS75} framework of relativization.

\subsection{Open Problems}
We think the most important open problem left in our paper is to study the algebrizing circuit complexity of $\MA_\E$. Can we strengthen \autoref{thm: intro MAE} to hold for all of $\MA_\E$ instead of just ``robust'' $\MA_\E$ algorithms? If not, can we prove an exponential size lower bound for $\MA_\E$ (potentially with one bit of advice) using algebrizing techniques?

Another open question is to prove a $\PP$ communication lower bound for $\XorMissingString$, which would imply an algebrization barrier to proving almost-everywhere circuit lower bounds for $\PEXP$. Such circuit lower bounds are known in the infinitely-often regime~\cite{BuhrmanFT98, DBLP:journals/tcs/Vinodchandran05}.

Finally, our work did not examine the regime of fixed-polynomial circuit lower bounds. Using algebrizing techniques, Santhanam~\cite{DBLP:journals/siamcomp/Santhanam09} proved that for every constant $k\ge 1$, $\MA/_1\not\subseteq \SIZE[n^k]$. Aaronson and Wigderson~\cite{DBLP:journals/toct/AaronsonW09} speculated that it might be possible to eliminate the advice bit and prove $\MA\not\subseteq \SIZE[n^k]$ using ``tried-and-true arithmetization methods,'' but there has been no progress in this direction. Is there an algebrizing barrier to proving $\MA\not\subseteq \SIZE[n^k]$?

%% file: preliminaries.tex
\section{Preliminaries}

\begin{definition}
    A \defn{search problem} $f$ over domain $\mathcal{X\times Y}$ and range $\mathcal{O}$ is defined using a relation $\mathcal{R}\subset \mathcal{(X\times Y)\times \mathcal{O}}$. For any input $(X,Y)\in \mathcal{X\times Y}$, the valid solutions for $f$ on $(X,Y)$ are those $s$ such that $(X,Y,s)\in \mathcal{R}$ (we say that \defn{$s$ is a solution to $f(X,Y)$}). We say that $f$ is a \defn{total problem}, if for any $(X,Y)$, there is at least one valid solution. \gpt{Notation: Prefer writing the domain as a Cartesian product of sets and ensure consistent use of calligraphic symbols for sets.}
\end{definition}

\subsection{Complexity Classes}

We assume familiarity with basic complexity classes such as $\P$, $\NP$, $\BPP$ and their exponential-time versions $\EXP$, $\NEXP$, $\sf BPEXP$. The reader is encouraged to consult standard textbooks~\cite{Arora-Barak, Goldreich-complexity} or the Complexity Zoo\footnote{\url{https://complexityzoo.net/}, accessed Nov 15, 2025.} for their definition.

\begin{definition}[{$\PostBPP$}]
    A promise problem $\Pi = (\Pi_{\rm yes}, \Pi_{\rm no})$ is in $\pr\PostBPP$ if there exist two polynomial-time algorithms $A, B$, taking an input $x\in \{0, 1\}^n$ and randomness $r \in \{0, 1\}^{\poly(n)}$ (they take the \emph{same} randomness), such that the following holds for every $x \in \Pi_{\rm yes} \cup \Pi_{\rm no}$.%
    \begin{itemize}
        \item If $x \in \Pi_{\rm yes}$, then
        $\Pr_r[A(x, r) = 1 \mid B(x, r) = 1] \ge 2/3$.
        \item If $x \in \Pi_{\rm no}$, then
        $\Pr_r[A(x, r) = 1 \mid B(x, r) = 1] \le 1/3$.
        \item $\Pr_r[B(x, r) = 1] > 0$.
    \end{itemize}
    We say that the algorithm \defn{postselects} on the event that $B(x, r) = 1$.
\end{definition}

\begin{definition}[{$\MA\cap\co\MA$}]
    A promise problem $\Pi = (\Pi_{\rm yes}, \Pi_{\rm no})$ is in $\pr(\MA\cap\co\MA)$ if there exist two polynomial-time algorithms (verifiers) $V_0,V_1$, taking an input $x\in \{0, 1\}^n$, a proof $\pi\in \{0, 1\}^{\poly(n)}$ and randomness $r \in \{0, 1\}^{\poly(n)}$ (they take \emph{different} randomness), such that the following holds for every $x \in \Pi_{\rm yes} \cup \Pi_{\rm no}$.
    \begin{itemize}
        \item If $x \in \Pi_{\rm yes}$, then $V_1$ accepts $x$ and $V_0$ rejects $x$.
        \item If $x \in \Pi_{\rm no}$, then $V_0$ accepts $x$ and $V_1$ rejects $x$.
    \end{itemize}
    Here, for $k=0,1$, we say that
    \begin{itemize}
        \item $V_k$ \defn{accepts} $x$, if there exists a proof $\pi$ such that $\Pr_r[V_k(x,\pi,r)=1]\ge 2/3$.
        \item $V_k$ \defn{rejects} $x$, if for any proof $\pi$, we have that $\Pr_r[V_k(x,\pi,r)=1]\le 1/3$.
    \end{itemize}
\end{definition}

%% file: multilinear.tex
\subsection{Algebrization}

In this section, we present some definitions regarding algebrization barriers. The definitions are in accordance with \cite{DBLP:journals/toct/AaronsonW09}, with slight modifications.

\begin{definition}[Oracle]
    An \defn{oracle} $A$ is a collection of Boolean functions $A_m:\{0,1\}^m\to \{0,1\}$, one for each $m\in \mathbb N$. Given a complexity class $\mathcal{C}$, by $\mathcal{C}^A$ we mean the class of languages decidable by a $\mathcal{C}$ machine that can query $A_m$ for any $m$ of its choice.
\end{definition}

We say that a separation $\mathcal{C}\not\subset \mathcal{D}$ does not algebrize, if there exist an oracle $A$ and a low-degree extension $\widetilde A$ of $A$, such that $\mathcal{C}^{\widetilde A}\subset \mathcal{D}^A$. Throughout this paper, we only consider multilinear extensions, which are the simplest class of low-degree extensions.%

\begin{definition}[Multilinear Extension over Finite Fields]
    Let $A_m:\{0,1\}^m\to\{0,1\}$ be a Boolean function, and let $\mathbb F$ be a finite field. The (unique) multilinear extension of $A_m$ over $\mathbb F$ is the linear function $\widetilde A_{m,\mathbb F}:\mathbb F^m\to \mathbb F$ that agrees with $A_m$ on $\{0,1\}^m$. \gpt{Terminology: A multilinear extension is a (unique) multilinear polynomial over $\mathbb F$—not a “linear function.” Consider rephrasing.} Given an oracle $A=(A_m)$, the \defn{multilinear extension} $\tilde A$ of $A$ is the collection of multilinear extensions $\widetilde A_{m,\mathbb F}$, one for each positive integer $m$ and finite field $\mathbb F$.

    Given a complexity class $\mathcal{C}$, by $\mathcal{C}^{\widetilde A}$ we mean the class of languages decidable by a $\mathcal{C}$ machine that can query $\widetilde A_{m,\mathbb F}$ for any $m,\mathbb F$ of its choice. Moreover, we assume that each query to $\widetilde A_{m,\mathbb F}$ takes $O(m\cdot \log|\mathbb F|)$ time.
\end{definition}

An important property of multilinear extensions is that algorithms having access to $\widetilde A$ can be simulated by communication protocols where Alice and Bob are each given half of $A$ as input. This reduces proving algebrization barriers to proving communication lower bounds. Formally, we have the following theorem, which is a simple corollary of \cite[Theorem 4.11]{DBLP:journals/toct/AaronsonW09}.

\begin{theorem}
    \label{thm:transfer}
    Let $A$ be an oracle, and let $A_0$ (resp.~$A_1$) be the subfunction of $A$ obtained by restricting the first bit to $0$ (resp.~$1$). Let $M$ be a deterministic algorithm that has oracle access to $\widetilde A$ and runs in time $T(|x|)$ when given $x$ as input. There exists a deterministic communication protocol $P$ in which Alice (resp.~Bob) is given $x$ and the function $A_0$ (resp.~$A_1$) as input, such that $P$ uses $O(T(|x|)^3)$ bits of communication, and agrees with $M(x)$ on any input $x$ and any oracle $A$.
\end{theorem}

We remark that the proof of \cite[Theorem 4.11]{DBLP:journals/toct/AaronsonW09} can be generalized to produce randomized communication protocols from randomized algorithms, $\MA$ communication protocols from $\MA$ algorithms, and so on.

%% file: missing_string_communication.tex
\section{An Infinitely-Often Algebrization Barrier for \texorpdfstring{$\PostBPE$}{PostBPE}}

In this section, we prove the following algebrization barrier:

{
    \def\ImplyMA{0}
    \def\Oracle{A}
    \ThmioPostBPPBarrier*
}

\subsection{\texorpdfstring{$\PostBPP$}{PostBPP} Communication Lower Bounds for \texorpdfstring{$\XorMissingString$}{XOR-Missing-String}}
\autoref{thm:io_postbpp_barrier} follows from the following communication lower bound for $\XorMissingString$:

\LemmaXORMSLB*

In this paper, $\PostBPP$ communication protocols are defined as follows:\footnote{There are many different definitions in the literature (see e.g., \cite{goos18landscape}). Here, we choose a simple definition that is best suited for proving the barrier.}

\begin{definition}[$\PostBPP$ communication protocols for search problems]\label{def: PostBPPcc}
    Let $f$ be a search problem over domain $\mathcal{X\times Y}$ and range $\mathcal{O}$. A \defn{$\PostBPP$ communication protocol} $\Pi$ for $f$ is defined as follows: Let $k$ be the length of the public randomness used by $\Pi$, and let $\{\Pi_r\}_{r\in \{0,1\}^k}$ be a set of deterministic communication protocols, each having communication complexity $c$. On input $(X,Y)\in \mathcal{X\times Y}$, protocol $\Pi$ first samples a uniformly random string $r\in \{0,1\}^k$, then simulates $\Pi_r$, whose output can be $\bot$ or any value from $\mathcal{O}$. The communication complexity of $\Pi$ is defined to be $k+c$.

    We say that $\Pi$ \defn{solves $f$ with error $\epsilon$}, if for any input $(X,Y)\in \mathcal{X\times Y}$,
    \begin{align*}
        \Pr_r\biggr[\Pi_r(X,Y)\text{ is a solution to }f(X,Y)\,\biggr|\,\Pi_r(X,Y)\ne \bot\biggr]\ge 1-\epsilon.
    \end{align*}

    We say that $\Pi$ is  \defn{pseudodeterministic}, if for any input $(X,Y)\in \mathcal{X\times Y}$, there exists some answer $s\in \mathcal{O}$, such that
    \begin{align*}
    \Pr\biggr[\Pi_r(X,Y)=s\,\biggr|\,\Pi_r(X,Y)\ne \bot\biggr]\ge 2/3.
    \end{align*}
\end{definition}

The main technical observation for our lower bound is that for any large enough rectangle $R$, no single answer $s\in \{0, 1\}^n$ will be correct for every input $(X, Y) \in R$. In fact, there is always a $2^{-\Theta(n)}$ fraction of inputs in $R$ on which $s$ is not a valid solution. 

To use this result in the lower bound, note that a $\PostBPP$ communication protocol can be seen as a distribution of labeled rectangles (i.e., approximate majority covers, as defined in \autoref{subsec:amc}). If a protocol has small communication complexity, then it must contain many large rectangles, so we can apply the aforementioned lower bound on individual rectangles, which implies that the protocol must have error $2^{-O(n)}$.

\subsubsection{A Lower Bound for One Rectangle}

In this section, we prove the main technical lemma for the communication lower bound. It is helpful to draw an analogy from the decision tree version of $\MissingString$: Suppose that the input is a list of $m$ $n$-bit strings. Whenever the decision tree only queries the input on $m-1$ bits, there always exists a string in the input list that is not queried. In this case, regardless of the decision tree's answer, the adversary can arbitrarily modify the value of that string and make the decision tree fail.

Here, we show that the $\XorMissingString$ problem has a similar property: If, conditioned on the communication history, the rectangle formed by the possible inputs is large (this corresponds to the decision tree making a small number of queries), then regardless of the protocol's answer, there always exists a significant portion of the rectangle, on which the protocol fails.

\begin{lemma}
    \label{lem:rect_error_prob_lb}
    Let $n\ge 1,1\le m<2^{n/2}$ be integers. Let $R\subseteq \{0,1\}^{nm}\times \{0,1\}^{nm}$ be a rectangle (i.e., $R$ is of the form $\mathcal{X\times Y}$, where $\mathcal{X,Y}\subseteq \{0,1\}^{nm}$) of size at least $2^{2nm-m+2}$. Let $s$ be any $n$-bit string. Then there exists some $R'$, which is a subrectangle of $R$ and has size at least $2^{-2n-2}|R|$, such that for any $(X,Y)\in R'$, $s$ is not a solution to $\XorMissingString$ on $(X,Y)$.
\end{lemma}

\begin{proof}
    Without loss of generality, we only have to prove the lemma for the case where $s=0^n$. When $s\ne 0^n$, we consider a different rectangle $R_s$, in which every input string $x$ for Alice in every input list $X\in \mathcal{X}$ is replaced by $x\oplus s$. By the definition of $\XorMissingString$, the lemma holds for $R$ and $s$ if and only if it holds for $R_s$ and $0^n$.

    To prove the lemma, we need to show that there exists a large enough subrectangle $R'$ in $R$, in which $0^n$ is never a solution to $\XorMissingString$. Recall that $R$ is equal to $\mathcal{X}\times \mathcal{Y}$, where $\mathcal{X}$ is a set of Alice's inputs and $\mathcal{Y}$ is a set of Bob's inputs. For any string $s\in\{0,1\}^n$, let $\mathcal{X}_s=\{X\in \mathcal{X},s\in X\}$ denote the subset of $\mathcal{X}$ that contains the string $s$; $\mathcal{Y}_s$ is defined similarly. 
    
    For any $s\in \{0,1\}^n$, consider the subrectangle $\mathcal{X}_s\times \mathcal{Y}_s\subseteq R$. Since any input list (Alice's or Bob's) in the subrectangle always contains $s$, $0^n$ is never a solution in $\mathcal{X}_s\times \mathcal{Y}_s$.

    In the remaining, we show that there exists some $s$ such that $\mathcal{X}_s\times \mathcal{Y}_s$ is large enough (i.e., has size $\ge 2^{-2n-2}|R|$) using a counting argument. If we can show this, then the lemma follows by letting $R'=\mathcal{X}_s\times \mathcal{Y}_s$. 

    Let $cx_s=|\mathcal{X}_s|$ denote the number of occurrences of string $s$ in $\mathcal{X}$. Note that, if some input $X\in \mathcal{X}$ contains multiple copies of $s$, $X$ is only counted once in $cx_s$. Let $x_1,\dots,x_{2^n}$ be all the $n$-bit strings, sorted in non-increasing order of $cx$. Similarly define $cy_s$ and $y_1, \dots, y_{2^n}$. We have the following lower bound on the number of occurrences of $x_i$:

    \begin{claim}
        \label{clm:occurrence_lb}
        For any $1\le i\le 2^n$, $cx_{x_i}$ is at least 
        \begin{align*}
            \frac{|\mathcal{X}|-(i-1)^m}{ 2^n-i+1}.
        \end{align*}
    \end{claim}

    \begin{claimproof}
        Let $k=\sum_{i'\ge i}cx_{x_{i'}}$, i.e., the sum of occurrences of $x_{\ge i}$. 
        
        Consider the inputs $X$ in $\mathcal{X}$ that only contain strings from $x_1,\dots,x_{i-1}$. The number of such $X$ is at most $(i-1)^m$. Therefore, at least $|\mathcal{X}|-(i-1)^m$ inputs in $\mathcal{X}$ contain at least one string from $x_{\ge i}$. Each of the $|\mathcal{X}|-(i-1)^m$ inputs contribute at least one to the sum $k$. That is,
        \begin{align*}
            k=\sum_{i'\ge i}cx_{x_{i'}}\ge |\mathcal{X}|-(i-1)^m.
        \end{align*}
        Since $x_i$ occurs at least as frequently as any string in $x_{i+1},\dots,x_{2^n}$, we have that $k\le cx_{x_i}\cdot (2^n-i+1)$. Therefore,
        \begin{align*}
            cx_{x_i}&\ge \frac k{2^n-i+1}\ge \frac{|\mathcal{X}|-(i-1)^m}{ 2^n-i+1}.\qedhere
        \end{align*}
    \end{claimproof}

    The same bound also applies to $cy_{y_i}$. Fix any $1\le i\le 2^n$ and consider the strings $x_1,\dots,x_i$ and $y_1,\dots,y_{2^n-i+1}$. By the pigeonhole principle, there must exist some string $s$ that appears in both $\{x_1,\dots,x_i\}$ and $\{y_1,\dots,y_{2^n-i+1}\}$. Using \autoref{clm:occurrence_lb}, we can show that $\mathcal{X}_s\times \mathcal{Y}_s$ is large, that is,
    \begin{align*}
        |\mathcal{X}_s\times \mathcal{Y}_s|&=cx_s\cdot cy_s \\
        &\ge cx_{x_i}\cdot cy_{y_{2^n-i+1}} &\text{(by definition of the lists $x,y$)}\\
        &\ge \frac{|\mathcal{X}|-(i-1)^m}{2^n-i+1}\cdot \frac{|\mathcal{Y}|-(2^n-i)^m}{i}.&\text{(by \autoref{clm:occurrence_lb})}
    \end{align*}
    Let $A=|\mathcal{X}|,B=|\mathcal{Y}|$. It now remains to show that, for any $A,B$ such that $1\le A,B\le 2^{nm}$ and $A\cdot B\ge 2^{2nm-m+2}$, there exists some $1\le i\le 2^n$ such that the rectangle chosen above is large enough. That is,
    \begin{align}
        \label{eq:size_lb}
        \frac{A-(i-1)^m}{2^n-i+1}\cdot \frac{B-(2^n-i)^m}{i}\ge 2^{-2n-2}\cdot A\cdot B.
    \end{align}
    Let $i^*\ge 1$ be the largest integer such that $2(i^*-1)^m\le A$. Note that $2\cdot 2^{nm}>A$, so we must have $i^*\le 2^n$. In this case, we have that $2(2^n-i^*)^m\le B$, since if not, then
    \begin{align*}
        &2(i^*)^m\cdot 2(2^n-i^*)^m>A\cdot B& (2(i^*)^m>A\text{ by definition of $i^*$}) \\
        \Rightarrow\;\;& (i^*\cdot (2^n-i^*))^m>A\cdot B/4\ge 2^{2nm-m} \\
        \Rightarrow\;\;& i^*\cdot (2^n-i^*)>2^{2n-1},
    \end{align*}
    which is impossible since $i\cdot (2^n-i)\le 2^{2(n-1)}$ for any $i$. Therefore, by plugging $i^*$ into \eqref{eq:size_lb}, we have that
    \begin{align*}
        &\frac{A-(i^*-1)^m}{2^n-i^*+1}\cdot \frac{B-(2^n-i^*)^m}{i^*} \\
        \ge{}& \frac{A/2}{2^n-i+1}\cdot \frac{B/2}{i^*}&(2(i^*-1)^m\le A\text{ and }2(2^n-i^*)^m\le B) \\
        \ge{}& \frac{A/2}{2^n}\cdot \frac{B/2}{2^n}=2^{-2n-2}\cdot A\cdot B.&&\qedhere
    \end{align*}%
\end{proof}

\autoref{lem:rect_error_prob_lb} can be extended to the case where, in addition to $(X,Y)$, Alice and Bob also receive some auxiliary input (which is independent of $(X,Y)$). This variant will be useful in \autoref{sec:ma_barrier}.

\begin{restatable}{corollary}{RectAux}
    
    \label{cor:rect_error_prob_lb_aux}
    Let $n\ge 1,1\le m<2^{n/2},a\ge 0$ be integers. Let $R\subseteq \{0,1\}^{nm+a}\times \{0,1\}^{nm+a}$ be a rectangle of size at least $2^{2nm+2a-m+2}$ (i.e., $R$ is of the form $X\times Y$ for some $X,Y\subset \{0,1\}^{nm+a}$). Each element in $R$ is interpreted as having the form $(X\circ t_x,Y\circ t_y)$, where $(X,Y)$ is an input to $\XorMissingString(n,m)$, and $t_x,t_y$ are two $a$-bit strings.
    
    Let $s$ be any $n$-bit string. There exists some $R'$, which is a subrectangle of $R$ and has size at least $2^{-2n-2}|R|$, such that for any $(X\circ t_x,Y\circ t_y)\in R'$, $s$ is not a solution to $\XorMissingString$ on $(X,Y)$.
\end{restatable}

\begin{proof}
    This corollary cannot be obtained by applying \autoref{lem:rect_error_prob_lb} in a black-box way. However, it is provable using the same techniques, i.e., by defining $\mathcal{X}_s=\{s\in X:(X\circ t_x)\in \mathcal{X}\}$ and $\mathcal{Y}_s=\{s\in Y:(Y\circ t_y)\in \mathcal{Y}\}$, and showing that $\mathcal{X}_s\times \mathcal{Y}_s$ is large for some $s$. The details are omitted.
\end{proof}

\subsubsection{Reducing \texorpdfstring{$\PostBPP$}{PostBPP} Communication to Approximate Majority Covers} 

\label{subsec:amc}

To use \autoref{lem:rect_error_prob_lb} in the lower bound proof, we use the characterization of $\PostBPP$ communication protocols by distributions of labeled rectangles called \emph{approximate majority covers}~\cite{Klauck2003}.

\begin{definition}[approximate majority covers]
    Let $f$ be a search problem over the domain $\mathcal{X\times Y}$ and range~$\mathcal{O}$. An \defn{approximate majority cover for~$f$} is a (multi) set of labeled rectangles~$\mathcal{AMC}=\{(R,s)\}$, where~$R\subseteq \mathcal{X\times Y}$ is a rectangle, i.e.,~$R$ is of the form~$A\times B$ for some~$A\subseteq\mathcal{X},B\subseteq\mathcal{Y}$, and~$s\in \mathcal{O}$ is the label of~$R$. The \defn{size} of~$\mathcal{AMC}$ is defined as the number of labeled rectangles.

    We say that the approximate majority cover \defn{solves $f$ with error $\epsilon$}, if for any $(X,Y)\in \mathcal{X\times Y}$, we have that
    \begin{align*}
        \Pr_{(R,s)\sim \mathcal{AMC}}[s\text{ is a solution to }f(X,Y)\;|\;(X,Y)\in R]\ge 1-\epsilon.
    \end{align*}
    That is, when we randomly sample from $\mathcal{AMC}$ a labeled rectangle $(R,s)$ that contains $(X,Y)$, the label $s$ is a solution with probability $2/3$.
\end{definition}

\begin{claim}
    \label{clm:protocol_to_amc}
    Let $f$ be a total search problem over domain $\mathcal{X\times Y}$ and range $\mathcal{O}$. If there exists a $\PostBPP$ communication protocol $\Pi$ that solves $f$ with error $\epsilon$ and has complexity $C$, then there exists an approximate majority cover that solves $f$ with error $\epsilon$ and has size $\le 2^{C}$.
\end{claim}

\begin{proof}
    Assume that $\Pi$ uses $k\le C$ bits of randomness. That is, $\Pi$ is the uniform distribution over $2^k$ deterministic communication protocols $\{\Pi_r\}$, where each protocol $\Pi_r$ has complexity $C-k$. 

    Construct an approximate majority cover $\mathcal{AMC}$ as follows: Initially, $\mathcal{AMC}$ is empty. Each deterministic protocol $\Pi_r$ partitions the input space into at most $2^{C-k}$ pairwise disjoint rectangles, where the answers of $\Pi_r$ in each rectangle are the same. For each such rectangle $R$ of $\Pi_r$, suppose the answer of $\Pi_r$ on $R$ is $v$, we add $(R,v)$ to $\mathcal{AMC}$ if and only if $v\ne \bot$.

    It is clear that the size of $\mathcal{AMC}$ is at most $2^C$. Moreover, for any $(X,Y)\in \mathcal{X\times Y}$, there is a one-to-one correspondence between the rectangles in $\mathcal{AMC}$ containing $(X,Y)$ and the protocols $\Pi_r$ for which $\Pi_r(X,Y)\ne \bot$. Therefore, the value of $\Pi(X,Y)$ is distributed the same as the label of a uniformly random rectangle in $\mathcal{AMC}$ that contains $(X,Y)$.
\end{proof}

\subsubsection{Putting Everything Together}

With the components in place, we can now prove \autoref{lem:xor_ms_lb}.

\LemmaXORMSLB*
\begin{proof}
    Let $\Pi$ be a $\PostBPP$ communication protocol that solves $\XorMissingString(n,m)$ with error $\le 2^{-5n}$, and let $C$ denote the complexity of $\Pi$. We convert $\Pi$ into an approximate majority cover $\mathcal{AMC}$ using \autoref{clm:protocol_to_amc}. We have that $\mathcal{AMC}$ solves $\XorMissingString(n,m)$ with error $\le 2^{-5n}$ and has size $2^{O(C)}$. We now show that the size of $\mathcal{AMC}$ is at least $2^{\Omega(m)}$, which proves the lemma. %

    Let $S=\sum_{(R,s)\in \mathcal{AMC}}|R|$ denote the total size of the rectangles. Since $\mathcal{AMC}$ is correct, each input $(X,Y)\in \{0,1\}^{nm}\times \{0,1\}^{nm}$ must be contained in at least one rectangle, which means that $S\ge 2^{2nm}$.

    Let
    \begin{align*}
        C_{\text{correct}}\coloneqq\sum_{(X,Y)\in \{0,1\}^{nm}\times \{0,1\}^{nm}}\sum_{(R,s)\in \mathcal{AMC},R\ni(X,Y)}[s\text{ is a solution to }\XorMissingString\text{ on }(X,Y)].
    \end{align*}
    
    Since $\mathcal{AMC}$ solves $\XorMissingString(n,m)$ with error $2^{-5n}$, we have that
    \begin{align*}
        &C_{\text{correct}}\\
        \ge {}&\sum_{(X,Y)\in \{0,1\}^{nm}\times \{0,1\}^{nm}}(1-2^{-5n})\cdot \sum_{(R,s)\in \mathcal{AMC},R\ni (X,Y)}1\tag{each $(X,Y)$ is correct w.h.p.} \\
        = {}&(1-2^{-5n})\cdot \sum_{(R,s)\in \mathcal{AMC}}\sum_{(X,Y)\in R}1 \\
        = {}&(1-2^{-5n})\cdot S. 
    \end{align*}

    On the other hand,
    \begin{align*}
        &C_{\text{correct}} \\
        ={}&\sum_{(R,s)\in \mathcal{AMC}}\sum_{(X,Y)\in R}[s\text{ is a solution to }\XorMissingString\text{ on }(X,Y)] \\
        ={}&\sum_{\substack{(R,s)\in \mathcal{AMC}\\|R|\ge 2^{2nm-m+2}}}\sum_{(X,Y)\in R}[s\text{ is a solution to }\XorMissingString\text{ on }(X,Y)] \\
        &+\sum_{\substack{(R,s)\in \mathcal{AMC}\\|R|< 2^{2nm-m+2}}}\sum_{(X,Y)\in R}[s\text{ is a solution to }\XorMissingString\text{ on }(X,Y)].
    \end{align*}
    It follows from \autoref{lem:rect_error_prob_lb} that the first summand is at most
    \begin{align*}
        \sum_{\substack{(R,s)\in \mathcal{AMC}\\|R|\ge 2^{2nm-m+2}}}(1-2^{-2n-2})\cdot |R|\le (1-2^{-4n})\cdot S.
    \end{align*}
    Let $|\mathcal{AMC}|$ denote the size of $\mathcal{AMC}$, then the second summand is at most
    \begin{align*}
        \sum_{\substack{(R,s)\in \mathcal{AMC}\\|R|< 2^{2nm-m+2}}}|R|\le{}& |\mathcal{AMC}|\cdot 2^{2nm-m+2}.\\
        \le{}& |\mathcal{AMC}|\cdot S/2^{m-2}. \tag{since $S\ge 2^{2nm}$}
    \end{align*}
    If $|\mathcal{AMC}| \le 2^{m/2}$, then the second summand is at most $2^{-m/2+2}\cdot S$, which is at most $2^{-8n}\cdot S$ since $m\ge 20n$. Summing everything together, we have that
    \begin{align*}
        C_{\text{correct}}\le (1-2^{-4n}+2^{-8n})\cdot S,
    \end{align*}
    which contradicts the previous bound of $C_{\text{correct}}\ge (1-2^{-5n})\cdot S$. Hence it must be the case that $|\mathcal{AMC}| > 2^{m/2}$.
\end{proof}

\paragraph{Lower bounds for pseudodeterministic protocols.} A consequence of \autoref{lem:xor_ms_lb} is that it also applies to \emph{pseudodeterministic} $\PostBPP$ communication protocols for $\XorMissingString$ with constant error (say $1/3$), as the error probability for such protocols can always be reduced to as small as possible by running it multiple times and outputting the majority answer. In contrast, the error probability of an arbitrary (non-pseudodeterministic) protocol cannot be amplified in general, since it is unclear how to verify whether a string is a valid solution to $\XorMissingString$.

\begin{claim}
    \label{clm:pseudo_det_error_reduction}
    Let $n\ge 1,1\le m<2^{n/2}$ be integers. Let $\Pi$ be a pseudodeterministic $\PostBPP$ communication protocol that solves $\XorMissingString(n,m)$ with error $1/3$ and has communication complexity $C$. Then for any $k\ge 1$, there exists a pseudodeterministic $\PostBPP$ communication protocol $\Pi_k$ that solves $\XorMissingString(n,m)$ with error $2^{-\Theta(k)}$ and has communication complexity $O(C\cdot k)$.
\end{claim}

By setting $k = \Theta(n)$ in the above lemma and applying \autoref{lem:xor_ms_lb}, we have:

\begin{corollary}\label{cor: LB for pseudodet PostBPP}
    Let $n\ge 1$ and $20n\le m < 2^{n/2}$ be integers. Any pseudodeterministic $\PostBPP$ communication protocol that solves $\XorMissingString(n, m)$ with error probability $\le 1/3$ requires communication complexity $\Omega(m/n)$.
\end{corollary}

\subsection{An Extension: \texorpdfstring{$\XorMissingString$}{XOR-Missing-String} Lower Bound for List-Solvers}\label{sec: list-solvers}

We prove a strengthened lower bound for $\XorMissingString$ against \emph{list-solvers}, which are protocols that output a list of strings, and are considered correct if \emph{at least one} string in the list is a solution. Although this strengthened lower bound is not needed for \autoref{thm:io_postbpp_barrier}, it is relevant in the context of \emph{almost-everywhere} circuit upper bounds and will be useful in the proof of \autoref{thm: intro BPE}.%

\begin{definition}
    Let $f$ be a search problem over domain $\mathcal{X\times Y}$ and range $\mathcal{O}$. We say that a communication protocol $\Pi$ is a \defn{list-solver} for $f$, if it outputs a list of answers (i.e., its output is in $\mathcal{O}^k$ for some $k$).

    For a $\PostBPP$ list-solver $\Pi$, we say that it \defn{list-solves $f$ with error $\epsilon$}, if for any input $(X,Y)\in\mathcal{X\times Y}$,
    \begin{align*}
        \Pr[\text{one of the answers in }\Pi(X,Y)\text{ solves }(X,Y)\mid \Pi(X,Y)\ne \bot]\ge 1-\epsilon.
    \end{align*}
\end{definition}

\begin{restatable}{lemma}{XorMulti}
    \label{lem:xor_ms_lb_multi}
    Let $n\ge 1,1\le m<2^{n/2},k$ be parameters, where $m\ge 20nk$. Let $Q$ be a $\PostBPP$ list-solver of  $\XorMissingString(n,m)$ that outputs $k$ candidate strings. If $Q$ list-solves $\XorMissingString(n,m)$ with error $\le 2^{-5nk}$, then the complexity of $Q$ is at least $\Omega(m)$.
\end{restatable}

It is natural to require that the error probability is at least $2^{-O(nk)}$, since this bound is achieved by a trivial protocol that outputs $k$ random strings.

The proof is similar to \autoref{lem:xor_ms_lb}: We use the lower bound on rectangles (\autoref{lem:rect_error_prob_lb}) to show that, for large enough rectangles, even when the protocol outputs $k$ candidate strings, there still exists a significant fraction of inputs in the rectangle on which all strings fail. We then use approximate majority covers to translate this result into a communication lower bound.

\begin{proof}

    We start by proving a stronger version of \autoref{lem:rect_error_prob_lb}.

    \begin{claim}
        \label{clm:rect_error_prob_lb_multi}
        Let $n\ge 1,1\le m<2^{n/2},k$ be parameters. Let $R\subseteq \{0,1\}^{nm}\times \{0,1\}^{nm}$ be a rectangle of size at least $2^{2nm-m+2+(2n+2)\cdot k}$. Let $s_1,\dots,s_k$ be any set of $k$ $n$-bit strings. Then there exists some $R'$, which is a subrectangle of $R$ and has size at least $2^{-(2n+2)\cdot k}|R|$, such that for any $(X,Y)\in R'$, none of $s_1,\dots,s_k$ is a solution to $\XorMissingString$ on $(X,Y)$.
    \end{claim}

    \begin{claimproof}
        The proof is by repeatedly applying \cref{lem:rect_error_prob_lb}. Given $R$ and $s_1,\dots,s_k$, we first apply \cref{lem:rect_error_prob_lb} on $R$ and $s_1$ to construct a subrectangle $R_1\subset R$, which has size at least $2^{-2n-2}|R|$, on which $s_1$ is never a solution. Next, since $|R_1|\ge 2^{-2n-2}|R|\ge 2^{2nm-m+2}$, we can apply \cref{lem:rect_error_prob_lb} on $R_1$ and $s_2$, and obtain another subrectangle $R_2\subset R_1$, on which $s_2$ is never the solution. This process can be repeated, and the end result is a subrectangle $R_k$ that refutes every $s_i$.
    \end{claimproof}

    Next, we convert $Q$ into an approximate majority cover. 
    
    Formally, define the $k$-fold $\XorMissingString(n,m)$ to be the search problem where, given an instance $(X,Y)$ of $\XorMissingString(n,m)$, the task is to output $k$ strings so that at least one of them solves $(X,Y)$. Since $Q$ is a protocol for this problem, we can use \autoref{clm:protocol_to_amc} to obtain an approximate majority cover $\mathcal{AMC}$, such that $\mathcal{AMC}$ has size at most $2^C$ (where $C$ is the complexity of $Q$), and the label corresponding to each rectangle is a list of $k$ strings. $\mathcal{AMC}$ satisfies that, for any $(X,Y)\in\{0,1\}^{nm}\times \{0,1\}^{nm}$,
    \begin{align}
        \label{equ:multi_solve}
        \Pr_{(R,s_1,\dots,s_k)\sim\mathcal{AMC}}[\text{one of }s_1,\dots,s_k\text{ solves }(X,Y)\mid (X,Y)\in R]\ge 1-2^{-5nk}.
    \end{align}
    
    Finally, we show that $\mathcal{AMC}$ has size at least $2^{\Omega(m)}$. Let

    \begin{align*}
        C_{\text{correct}}=\sum_{(X,Y)}\sum_{(R,s_1,\dots,s_k)\in \mathcal{AMC},(X,Y)\in R}[\text{one of }s_1,\dots,s_k\text{ solves }(X,Y)].
    \end{align*}

    Let $S=\sum_{(R,s_1,\dots,s_k)\in \mathcal{AMS}}|R|$ denote the total size. The following computations are similar to the proof of \autoref{lem:xor_ms_lb}, where similar steps are simplified.
    
    By \eqref{equ:multi_solve}, we have that
    \begin{align*}
        C_{\text{correct}}\ge (1-2^{-5nk})\cdot S.
    \end{align*}

    On the other hand, using \autoref{clm:rect_error_prob_lb_multi}, we have that 
    \begin{align*}
        C_{\text{correct}}&\le \sum_{R\in \mathcal{AMC}:|R|\ge 2^{2nm-m+2+(2n+2)\cdot k}}(1-2^{-(2n+2)\cdot k})\cdot |R|+\sum_{R\in \mathcal{AMC}:|R|\le 2^{2nm-m+2+(2n+2)\cdot k}}|R| \\
        &\le (1-2^{-(2n+2)\cdot k})\cdot S+2^C\cdot 2^{2nm-m+2+(2n+2)\cdot k}.
    \end{align*}
    Since $\mathcal{AMC}$ solves every input, we must have that $S\ge 2^{2nm}$. Therefore, the second term is at most $2^C\cdot 2^{-m+2+(2n+2)\cdot k}\cdot S$, which is $2^{-\Omega(m)}\cdot S$ since $m\ge 20nk$. \gpt{Clarity: Replace “which is $-\Omega(m)$” with “which is $2^{-\Omega(m)}\cdot S$.”} We thus have that
    \begin{align*}
        (1-2^{(2n+2)\cdot k})\cdot S+2^C\cdot 2^{-\Omega(m)}\cdot S\ge (1-2^{-5nk})\cdot S.
    \end{align*}
    For this to hold, it must be the case that $C=\Omega(m)$.
\end{proof}

\subsection{Proving the Barrier}

{
    \def\ImplyMA{0}
    \def\Oracle{A}
    \ThmioPostBPPBarrier*
}
\def\GOOD{\mathcal{GOOD}}
\begin{proof}[Proof of \autoref{thm:io_postbpp_barrier}]
    Let $M_1, M_2, \dots$ be a syntactic enumeration of $\PostBPE$ algorithms each running in time $2^n$ (that is, each $M_i$ may or may not satisfy the $\PostBPP$ promise). If we can prove a barrier for all such algorithms, then we can also prove a barrier for all algorithms running in time $2^{O(n)}$ by a padding argument. We assume that each machine $M_i$ appears infinitely many times in the list $(M_i)_{i\in \mathbb{N}}$.
    
    Let $n_1$ be a large enough constant and set $n_i = 4^{n_{i-1}}$ for every $i > 1$. Let $m_i = n_i^{20}$ for every $i\ge 1$. Recall that logarithms are always base $2$. The oracle $A$ that we construct will have the following structure: On input strings whose lengths are not of the form $2\log m_i$, the value of $A$ will always be zero. For every $i\ge 1$, the truth table of $A$ on input length $2\log m_i$ will encode an instance $(X_i, Y_i)$ of $\XorMissingString(n_i, m_i)$. More specifically, the truth table of $A\cap \{0, 1\}^{2\log m_i}$ restricted to inputs whose first bit is $0$ (resp.~$1$) will be equal to $X_i$ (resp.~$Y_i$) padded with zeros. Note that the length of $(X_i, Y_i)$ is $2m_i\cdot n_i \le 2^{2\log m_i}$.
    
    For every $i\ge 1$, the instance $(X_i, Y_i)$ is designed to diagonalize against $M_i^{\widetilde{A}}$. More precisely, we want the following property to hold: let $\GOOD_i$ be the set of inputs $x \in \{0, 1\}^{\log n_i}$ such that the execution of $M_i^{\widetilde{A}}(x)$ satisfies the semantics of $\PostBPP$, then there exists a \emph{non-solution} $\alpha \in \{0, 1\}^{n_i}$ of the $\XorMissingString$ instance $(X_i, Y_i)$ such that $\alpha_x = M_i^{\widetilde{A}}(x)$ for every $x \in \GOOD_i$. Note that $\alpha = (X_i)_a \oplus (Y_i)_b$ for some $a, b \in [m_i]$, hence by hardwiring $a$ and $b$ we can construct an $A$-oracle circuit of size $O(\log m_i) \le O(\log n_i)$ whose truth table is equal to $\alpha$, and it follows that the $\pr\PostBPE$ problem defined by $M_i^{\widetilde{A}}$ on input length $\log n_i$ can be computed by linear-size $A$-oracle circuits. Therefore, it suffices to satisfy the aforementioned property.

    Note that, when considering the truth table of $M^{\widetilde{A}}_i$ on input length $\log n_i$, the algorithm $M^{\widetilde{A}}_i$ can only access $\widetilde{A}$ on inputs of length $\le n_i$. Therefore, we only need to show that $M^{\widetilde{A}_{\le n_i}}_i$ fails to solve $(X_i,Y_i)$. Since $2\log m_{i+1}> n_i$, this implies that the algorithm only sees $(X_{\le i},Y_{\le i})$.

    We construct our oracle inductively. Suppose that we have constructed $(X_j, Y_j)$ for every $j < i$, and we need to construct $(X_i, Y_i)$. Let $P_i$ be the following $\PostBPP$ communication protocol that tries to solve $\XorMissingString(n_i, m_i)$. Given an instance $(X, Y)$, define an oracle $A$ such that $(X_{<i}, Y_{<i})$ are the previously fixed values and $(X_i, Y_i) = (X, Y)$. Then, Alice and Bob output the ``truth table'' of $M_i^{\widetilde{A}_{\le n_i}}$ on inputs of length $\log n_i$, denoted as ${\bf tt} \in \{0, 1\}^{n_i}$: for every $x \in \{0, 1\}^{\log n_i}$, Alice and Bob simulate the computation of $M_i^{\widetilde{A}_{\le n_i}}(x)$ using (the $\PostBPP$ version of) \autoref{thm:transfer} repeatedly for $(1000n_i+1)$ times, and output the majority outcome as the $x$-th bit of ${\bf tt}$. (We use boldface to emphasize that ${\bf tt}$ is a random variable). Let $\GOOD_i(X, Y)$ denote the set of inputs $x \in \{0, 1\}^{\log n_i}$ such that the computation of $M_i^{\widetilde{A}_{\le n_i}}(x)$ satisfies $\PostBPP$ promise, and let $f_i : \GOOD_i(X, Y) \to \{0, 1\}$ denote the promise problem defined by $M_i^{\widetilde{A}_{\le n_i}}$:
    \[\GOOD_i^b(X, Y) = \mleft\{x \in \{0, 1\}^{\log n_i}: \Pr\mleft[M_i^{\widetilde{A}_{\le n_i}}(x) = b~\biggl|~M_i^{\widetilde{A}_{\le n_i}}(x) \ne \bot\bigr.\mright] \ge 2/3\mright\},\]
    \[\GOOD_i(X, Y) = \GOOD_i^0(X, Y) \cup \GOOD_i^1(X, Y),\]
    \[f_i(x) = \begin{cases}0 & \text{if }x \in \GOOD_i^0(X, Y);\\1 & \text{if }x \in \GOOD_i^1(X, Y).\end{cases}\]

    We say that ${\bf tt}$ is \emph{consistent} with $f_i$ if for every $x \in \GOOD_i(X, Y)$, ${\bf tt}_x = f_i(x)$. Since the computation of $M_i^{\widetilde{A}_{\le n_i}}(x)$ is repeated $(1000n_i+1)$ times, the probability that ${\bf tt}$ is consistent with $f_i$ is at least $1 - 2^{-10n_i}$. On the other hand, note that the behavior of $P_i$ is the same as some $\PostBPP$ algorithm that has oracle access to $\widetilde{A}$ and runs in time $O(n_i^3)$ (since it considers $n_i$ possible inputs, simulates $M_i$ for $O(n_i)$ times on each of them, and each simulation takes $O(n_i)$ time), hence \autoref{thm:transfer} implies that $P_i$ can be implemented by a $\PostBPP$ communication protocol of complexity $O(n_i^{9})=o(m_i)$. It follows from \autoref{lem:xor_ms_lb} that $P_i$ cannot solve $\XorMissingString(n_i, m_i)$ with success probability more than $1-2^{-5n_i}$; in other words, there exists some $(X, Y)$ such that ${\bf tt}$ is a non-solution of $(X, Y)$ w.p.~at least $2^{-5n_i}$. By a union bound, there is a non-zero probability that both of the following hold simultaneously: ${\bf tt}$ is consistent with $f_i$ and ${\bf tt}$ is a non-solution of $(X, Y)$. This means that there is a non-solution $\alpha \in \{0, 1\}^{n_i}$ of $(X, Y)$ such that for every $x \in \GOOD_i(X, Y)$, $M_i^{\widetilde{A}_{\le n_i}}(x) = \alpha_x$. Setting $(X_i, Y_i) = (X, Y)$, this establishes the property we want for $i$, and we can continue our construction for $i+1$.
\end{proof}

%% file: bpp_ae.tex
\section{Alternative Almost-Everywhere Algebrization Barrier for \texorpdfstring{$\BPE$}{BPE}}

In \cite{DBLP:journals/toct/AaronsonW09}, the authors proved almost-everywhere algebrization barriers for many common complexity classes such as $\NEXP$ and $\BP\EXP$. However, their barriers were constructed using multiquadratic extensions, which are more difficult to manipulate than multilinear extensions, and are arguably less clean. Barriers based on multiquadratic extensions also fail to imply \emph{affine relativization} barriers in the sense of~\cite{DBLP:journals/toct/AydinliogluB18}. It was unknown whether these barriers can be proven using only communication complexity. In this section, we use the techniques developed for $\XorMissingString$ to show that these almost-everywhere algebrization barriers can indeed be proven using communication complexity only.

{
    \def\Oracle{A}
    \ThmBPEBarrier*
}
\subsection{Reducing to Communication Lower Bounds for \texorpdfstring{$\XorMissingString$}{XOR-Missing-String}}

We first show that certain communication lower bounds suffice for proving our barrier result. Roughly speaking, we consider a setting where there are many inputs $\calI_1, \calI_2, \dots, \calI_\ell$ and many protocols $\calP_1, \calP_2, \dots, \calP_\ell$, every protocol $\calP_i$ receives every input $(\calI_1, \dots, \calI_\ell)$, and the goal of $\calP_i$ is to solve the problem on input $\calI_i$. We want to prove lower bounds of the following form: there is a sequence of bad inputs $(\calI_1, \dots, \calI_\ell)$ such that \emph{no protocol $\calP_i$ could solve $\calI_i$ correctly, even after seeing the inputs of other protocols}. This setting captures win-win analyses: for example, if $\calP_i$ fails to solve $\calI_i$, then $\calI_i$ might contain useful information that helps $\calP_j$ solve $\calI_j$. In fact, this is indeed what happened in many win-win analyses in complexity theory: see e.g., \cite[Section 1.4.2]{CHR24}, \cite[Section 5.2]{LuORS24}, and \autoref{sec: BFT as algebraic query algorithms} of this paper.

We set the following parameters. Let $n_1=2^{C}$ for some sufficiently large constant $C$ and $n_i = n_{i-1}^4$ for every $i > 1$. For every $i\ge 1$, let $m_i = n_i^{100}$, the $i$-th input will be an instance of $\XorMissingString(n_i, m_i)$. We also define a sequence of error thresholds $\{p_i\}_{i\in\N}$: let $p_1 = 0.6$ and $p_i = p_{i-1} + 1/(10i^2)$ for every $i > 1$, then $p_i < 0.8$ for every $i$. Looking ahead, it will be useful in the proof of \autoref{lemma: finite barrier} that there is a non-trivial gap (of $1/(10i^2)$) between each $p_i$ and $p_{i+1}$.

We need the following communication lower bound:

\begin{theorem}\label{lemma: finite barrier}
    Let $t\ge 1$ be an integer and $\{Q_{i, j}\}_{1\le j\le i\le t}$ be a set of randomized communication protocols satisfying the following:
    \begin{itemize}
        \item In each $Q_{i, j}$, Alice receives $\vec{X} := (X_1, \dots, X_t)$ and Bob receives $\vec{Y} := (Y_1, \dots, Y_t)$, where each $(X_i, Y_i)$ is an instance of $\XorMissingString(n_i, m_i)$.
        \item The communication complexity of each $Q_{i, j}$ is at most $n_i^4$ and each $Q_{i, j}$ outputs a string of length $n_i$.
    \end{itemize}
    Then there exists a sequence of inputs $\{(X_i, Y_i)\}_{i\le t}$ such that every $Q_{i, j}$ fails to solve the instance $(X_i, Y_i)$. More formally, for every $1\le j\le i\le t$:
    \begin{itemize}
        \item either no string is outputted by $Q_{i, j}(\vec{X}, \vec{Y})$ with probability $>p_t$; or
        \item the (unique) string outputted with probability $>p_t$ is not a solution to $\XorMissingString$ on the input $(X_i, Y_i)$.
    \end{itemize}
\end{theorem}
 \gpt{Label hygiene: The label key contains a space (\texttt{lemma: finite barrier}). Consider \texttt{lemma:finite-barrier} to avoid fragile references.}

\begin{proof}[Proof of \autoref{thm: intro BPE} from \autoref{lemma: finite barrier}]
    Similar to the proof of \autoref{thm:io_postbpp_barrier}, our oracle encodes a sequence of $\XorMissingString$ instances. An infinite sequence of instances $\{(X_i, Y_i)\}_{i\in \N}$ where each $(X_i, Y_i)$ is an instance of $\XorMissingString(n_i, m_i)$ corresponds to an oracle $A$ defined as follows:\begin{itemize}
        \item For every $i\ge 1$, the truth table of $A$ on input length $2\log m_i$ encodes the instance $(X_i, Y_i)$. More precisely, the truth table of $A\cap \{0, 1\}^{2\log m_i}$ restricted to inputs whose first bit is $0$ (resp.~$1$) is equal to $X_i$ (resp.~$Y_i$) padded with zeros. Note that the length of $(X_i, Y_i)$ is $2m_i\cdot n_i \le 2^{2\log m_i}$.
        \item If the input length is not of the form $2\log m_i$, then $A$ always returns $0$.
    \end{itemize}
    Oracles of this form are called \defn{well-structured oracle}. When the oracle $A$ corresponds to a finite sequence of $t$ instances $\{(X_i, Y_i)\}_{i\le t}$ (that is, $(X_i, Y_i)$ is the all-zero string for every $i > t$), we say that $A$ is a \defn{truncated oracle at level $t$}. Note that we consider oracles that are much ``denser'' than those considered in \autoref{thm:io_postbpp_barrier}, in the sense that the lengths of adjacent $\XorMissingString$ instances encoded in the oracle are much closer to each other. (Recall that we set $n_i = n_{i-1}^4$ in this section but we set $n_i = 4^{n_{i-1}}$ in \autoref{thm:io_postbpp_barrier}.) 

\paragraph{Converting algorithms to communication protocols.} Let $M_1,M_2,\dots$ be a syntactic enumeration of $\BPE$ algorithms, each running in time $2^n$. For any integer $i\ge 1,j\ge 1$, let $P_{i,j}$ be the communication protocol that outputs the truth table of $M_j^{\widetilde A}$ on inputs of length less than $\log n_i$:
\begin{itemize}
    \item Alice (resp.~Bob) receives as inputs the subfunction $A_0$ (resp.~$A_1$) of a well-structured oracle $A$, where $A_0$ (resp.~$A_1$) denotes the subfunction of $A$ where the first bit is restricted to $0$ (resp.~$1$). Note that since $M_j$ runs in $2^{\log n_i} = n_i$ time, for some large enough $b_i := O(\log n_i)$, it only has access to the instances $(X_j, Y_j)$ for $j\le b_i$. Hence one can view $P_{i, j}$ as a communication protocol where Alice (resp.~Bob) receives $(X_1, \dots, X_{b_i})$ (resp.~$(Y_1, \dots, Y_{b_i})$) as inputs.
    \item Then, for every $\ell < \log n_i$, Alice and Bob simulate the execution of $M_j^{\widetilde A}$ on inputs of length $\ell$ and obtains a truth table $tt_\ell$. They output the concatenation of $tt_\ell$ for every $\ell < \log n_i$; this is a bit-string of length $n_i - 1$, and we append a bit $0$ at the end to make the output length equal to $n_i$.
    \item On each input of $M_j^{\widetilde{A}}$, Alice and Bob simulate the execution of $M_j^{\widetilde A}$ for $\Theta(n_i)$ times and output the majority answer. This means that when $M_j^{\widetilde A}$ satisfies the semantics of a $\BPP$ algorithm, $P_{i,j}$ successfully computes the truth table with probability $\ge 1 - 2^{-\Omega(n_i)}$.
\end{itemize}
Using (a generalization of) \autoref{thm:transfer}, $P_{i,j}$ can be implemented by a randomized communication protocol with complexity $O(n_i^3) < n_i^4$ if $n_i$ is large enough.

\paragraph{Constructing the oracle.} For each $k\in \N$, define $S_k$ to be the set of truncated oracles at level $b_k$ (i.e., instance sequences $\{(X_i, Y_i)\}_{i \le b_k}$) such that for every $1\le j\le i\le k$, the protocol $P_{i,j}$ fails to solve $(X_i, Y_i)$ with probability threshold $0.8$. Using \autoref{lemma: finite barrier}, we can show that

    \begin{claim}\label{claim: Sk nonempty}
        For any $k\ge 1$, $S_k$ is nonempty.
    \end{claim}

    \begin{claimproof}
        This is a consequence of \autoref{lemma: finite barrier}. More specifically, we use $b_k$ as the parameter $t$ in \autoref{lemma: finite barrier}, let $Q_{i,j}=P_{i,j}$ for every $i\le k$, and $Q_{i,j}$ be a trivial protocol for $i>k$. Note that each $Q_{i,j}$ is a communication protocol with complexity $\le n_{i}^4$. Hence, by invoking \autoref{lemma: finite barrier}, we obtain a truncated oracle $A$ at level $b_k$, on which every protocol $P_{i,j}$ fails with probability threshold $p_{b_k} < 0.8$. It follows that $A\in S_k$.
    \end{claimproof}

    Let $t_1\le t_2$. For instance sequences (i.e., truncated oracles) $A^{(1)} = \{(X_i^{(1)}, Y_i^{(1)})\}_{i\le t_1}$ and $A^{(2)} = \{(X_i^{(2)}, Y_i^{(2)})\}_{i\le t_2}$, we say that $A^{(1)}$ is a \defn{prefix} of $A^{(2)}$ if $(X_i^{(1)}, Y_i^{(1)}) = (X_i^{(2)}, Y_i^{(2)})$ for every $i\le t_1$; that is, they agree on the first $t_1$ instances. We say that a truncated oracle $A\in S_k$ \defn{extends infinitely} if there exist infinitely many oracles in $\bigcup_{k\ge i}S_k$ for which $A$ is a prefix. The following two properties of $\{S_k\}$ are easy to see:
    \begin{itemize}
        \item For any $A\in S_i$ and any $j<i$, there exists an oracle in $S_{j}$ that is a prefix of $A$.
        \item The empty oracle (corresponding to the empty sequence of instances) extends infinitely (this follows from \autoref{claim: Sk nonempty}).
    \end{itemize}

    We construct a sequence $\{(X_i, Y_i)\}_{i\in\N}$ as follows. Suppose we have constructed a finite sequence $A_\ell = \{(X_i, Y_i)\}_{i\le \ell}$, we maintain the invariant that $A_\ell$ extends infinitely. We start from $\ell = 0$ and $A_\ell$ being the empty oracle. To construct $(X_{\ell + 1}, Y_{\ell + 1})$, we simply find any $(X_{\ell + 1}, Y_{\ell + 1})$ that would maintain our invariant that $\{(X_i, Y_i)\}_{i\le \ell + 1}$ extends infinitely. Since there are only finitely many options for $(X_{\ell + 1}, Y_{\ell + 1})$, such a choice of $(X_{\ell + 1}, Y_{\ell + 1})$ always exists.

    Let $A$ denote the well-structured oracle corresponding to the sequence $\{(X_i, Y_i)\}_{i\in\N}$ and we now prove that $\BPE^{\widetilde{A}} \subseteq \SIZE^A[O(n)]$. Consider a machine $M_j^{\widetilde{A}}$ and a sufficiently large input length $n$. Let $i$ be the integer such that $n \in [\log n_{i-1}, \log n_i)$. In this case, we have that $n \ge (\log n_i)/4$. Since $n$ is sufficiently large, we may assume that $i \ge j$. Note that the behavior of $P_{i,j}$ on $A$ is the same as its behavior on $A_{b_i} = \{(X_{i'}, Y_{i'})\}_{i' \le b_i}$ (as $M_j$ on input length $\log n_i$ can only access these inputs). Since $A_{b_i} \in S_i$, we have that $P_{i, j}$ fails to solve the instance $(X_i, Y_i)$ with probability threshold $0.8$. There are two cases:
    \begin{itemize}
        \item Either $M_j^{\widetilde{A}}$ does not satisfy the $\BPP$ promise. In this case, we are done.
        \item Or $M_j^{\widetilde{A}}$ satisfies the $\BPP$ promise but the unique string outputted by $P_{i,j}$ is a string of the form $x \oplus y$ where $x\in X_i$ and $y\in Y_i$. In this case, there is an $A$-oracle circuit of size $O(\log m_i)=O(\log n_i)$ that outputs the truth table of $M_j^{\widetilde A}$ on inputs of length $<\log n_i$. From this, we can construct a circuit which agrees with $M_j^{\widetilde A}$ on inputs of length $n$ and has size $O(\log m_i)=O(n)$.\qedhere
    \end{itemize}
\end{proof}

\subsection{Proof of \autoref{lemma: finite barrier}}

We prove \autoref{lemma: finite barrier} by induction. The base case where $t = 1$ follows from \autoref{cor: LB for pseudodet PostBPP}. Now assume that for some $t>1$, \autoref{lemma: finite barrier} holds for $t-1$ but not for $t$, and $\{Q_{i, j}\}_{1\le j\le i\le t}$ is a set of protocols that witnesses the failure of \autoref{lemma: finite barrier}. We show that $\{Q_{i, j}\}$ can be transformed into a single $\PostBPP$ protocol $Q$ for list-solving $\XorMissingString(n_t, m_t)$ that is impossibly efficient in the sense that it contradicts our communication lower bound (\autoref{lem:xor_ms_lb_multi}).

Let $(X_t,Y_t)$ be an input for $\XorMissingString(n_t,m_t)$. For $1\le j\le i\le t$, let $Q'_{i,j}$ be the protocol that takes as input a truncated oracle at level $t-1$, combines it with $(X_t,Y_t)$ to form a truncated oracle at level $t$, and then simulates the execution of $Q_{i,j}$ on the combined oracle. Consider the protocols $Q'_{i, j}$ for $1\le j \le i \le t - 1$. Since \autoref{lemma: finite barrier} holds for $t-1$, there must exist a truncated oracle $A$ at level $t-1$ on which all these protocols fail. That is, when we combine this oracle $A$ with $(X_t,Y_t)$ to form an oracle $A'$ and run the original protocols $\{Q_{i,j}\}$ on $A'$, every protocols $Q_{i, j}$ with $i\le t-1$ would fail. However, one of the protocols $Q_{i,j}$ must succeed on $A'$. In this case, the protocol that succeeds must have $i=t$, which means that its output solves $(X_t,Y_t)$.

The above discussion naturally leads to the following $\PostBPP$ communication protocol for list-solving $\XorMissingString(n_t, m_t)$: randomly guess a truncated oracle $A$ at level $t-1$, use postselection to ensure that $Q'_{i,j}$ fails for every $1\le j\le i\le t-1$, and simulate the protocols $Q_{j, t}$ to list-solve the instance $(X_t, Y_t)$. We remark that the combined protocol $Q$ needs to use \emph{postselection} even though the original protocols $Q_{i, j}$ do not; nevertheless, we can still reach a contradiction as \autoref{lem:xor_ms_lb_multi} applies to $\PostBPP$ communication protocols as well.

See \autoref{pro:postbpp_xor_ms} for a formal description of the $\PostBPP$ communication protocol $Q$.

Let $K$ denote the number of bits encoded in a truncated oracle of size $t-1$, formally, $K=\sum_{i<t}2m_in_i$. Note that $K=O(t\cdot n_{t-1}\cdot m_{t-1})=O(n_t^{30})\ll m_t$.

\begin{algorithm2e}[ht]
  \caption{The protocol $Q$}\label{pro:postbpp_xor_ms}
    \SetKwProg{Fn}{Q}{:}{}

    \KwIn{Alice gets $X_t$ and Bob gets $Y_t$, where $(X_t,Y_t)$ is an input to $\XorMissingString(n_t,m_t)$.}
    \KwOut{A list of $t$ $n_t$-bit strings, or $\bot$.}

    \For{$i\gets 1$ \KwTo $t-1$}{
        $(X_i,Y_i)\gets \text{uniformly random input of $\XorMissingString(n_i,m_i)$}$\;
    }
    $A\gets \text{the truncated oracle at level $t$ that encodes } (X_{\le t},Y_{\le t})$\;

    $p\gets (p_{t-1}+p_t)/2$\;
    
    \For{$i\gets 1$ \KwTo $t-1$}{\label{line: check start}
        \For(\tcp*[f]{Check whether $Q_{i,j}$ fails}){$j\gets 1$ \KwTo $i$}{
            Simulate the execution of $Q_{i,j}$ on input $A$, \\ \quad repeat for $K\cdot t^{10}$ times using independent randomness\;
            $s\gets \text{the empirical majority answer}$\;
            $p_{\sf emp}\gets \text{empirical probability of outputting }s$\;
            \uIf{$p_{\sf emp} > p$ and $s$ is a solution to $(X_i,Y_i)$}{
                \Return $\bot$\;\label{line: check end}
            }
        }
    }

    \For(\tcp*[f]{If all protocols $Q_{<t,j}$ fail, run $Q_{t,\le t}$ for answers}){$j\gets 1$ \KwTo $t$}{
        Simulate the execution of $Q_{t,j}$ on input $A$, \\ \quad repeat for $(n_t)^2$ times using independent randomness\;
        $s_j\gets \text{the empirical majority answer}$\;
    }
    \Return{$\langle s_1,\dots,s_t\rangle$}.
\end{algorithm2e}

\paragraph{Complexity of $Q$.} Recall that the complexity of a $\PostBPP$ communication protocol is defined as $k+c$, where $k$ is the length of the public randomness and $c$ is the number of bits communicated (\autoref{def: PostBPPcc}). The complexity of $Q$ consists of the following parts:

\begin{itemize}
    \item Sampling $(X_{<t},Y_{<t})$: This costs at most $K=O(n_t^{30})$ random bits.
    \item Simulating the protocols $Q_{i,j}$ for each $1\le j\le i\le t$: This requires communicating at most $t^2\cdot (n_t)^4\cdot K\cdot t^{10}=O(n_t^{46})$ bits.
\end{itemize}

Hence, the complexity of $Q$ is $O(n_t^{46}) \ll m_t$.

We remark that $Q$ also needs to check whether the outputs of $Q_{i,j}$ is a solution of $(X_i, Y_i)$ for each $1\le j\le i \le t-1$. However, since $(X_{<t},Y_{<t})$ is sampled using public randomness, this costs no communication.

\paragraph{Success probability of $Q$.} For a fixed input $(X_t,Y_t)$, let $E_{t-1}$ denote the following event: for the oracle $A$ constructed in \autoref{pro:postbpp_xor_ms}, all the protocols $Q_{i,j}$ with $1\le j\le i\le t-1$ fail on $A$ with probability threshold $p_{t-1}$. Let $E_t$ denote the event that all these protocols (i.e., $\{Q_{i,j}\}_{1\le j\le i\le t-1}$) fail with probability threshold $p_t$. Note that $E_{t-1}$ implies $E_{t}$, which means that $\Pr[E_{t-1}]\le \Pr[E_t]$.%

In the following, we lower bound $\Pr[Q\text{ is correct}\mid Q(X,Y)\ne \bot]$ by considering the behavior of $Q$ under three events:
\begin{enumerate}
    \item \textbf{Bad case:} Conditioned on $\lnot E_{t}$, $Q$ outputs $\bot$ with high probability (\autoref{clm:semantics_check}). Therefore, this case only negligibly affects $\Pr[Q\text{ is correct}\mid Q(X,Y)\ne \bot]$.
    \item \textbf{Good case:} Conditioned $E_{t-1}$, with high probability, $Q$ does not output $\bot$ (\autoref{clm:semantics_check}), and is correct (\autoref{clm:correctness_bound}). We also show that $\Pr[E_{t-1}]$ is large (\autoref{clm:at_least_one}), so that this case has a big contribution to $\Pr[Q\text{ is correct}\mid Q(X,Y)\ne \bot]$.
    \item \textbf{Intermediate case:} Conditioned on $E_{t}\land \lnot E_{t-1}$, $Q$ behaves abnormally, in that we cannot bound $\Pr[Q(X,Y)\ne \bot\mid E_{t}\land \lnot E_{t-1}]$. However, we know that conditioned on $E_{t}\land \lnot E_{t-1}$ and $Q(X,Y)\ne \bot$, $Q$ is correct with high probability (\autoref{clm:correctness_bound}). Therefore, regardless of $\Pr[Q(X,Y)\ne \bot\mid E_{t}\land \lnot E_{t-1}]$, this case does not hurt $\Pr[Q\text{ is correct}\mid Q(X,Y)\ne \bot]$.
\end{enumerate}

Let $S$ denote the event that, when running the simulation, one of the output strings $s_1,\dots,s_t$ solves $(X_t,Y_t)$. We define $S$ in such a way that $S$ is independent from whether $Q(X_t,Y_t)$ outputs $\bot$. That is, it is possible that although $Q(X_t,Y_t)$ outputs $\bot$, had $Q$ performed the simulation, it would have solved $(X_t,Y_t)$. In this case, we say that $S$ holds.

We start by proving that the good case happens with significant probability:

\begin{claim}
    \label{clm:at_least_one}
    $\Pr[E_{t-1}]\ge 2^{-K}$.
\end{claim}

\begin{claimproof}
    Consider the protocols $\{Q'_{i,j}\}$ $(1\le j\le i<t)$, where $Q'_{i,j}$ takes a truncated oracle $A'$ of size $t-1$, combines it with $(X_t,Y_t)$, and then simulates $Q_{i,j}$ on $A$. Since \autoref{lemma: finite barrier} holds for $t-1$, applying it on $\{Q'_{i,j}\}$ implies that there exists at least one truncated oracle $A'$ of size $t-1$, on which all $Q'_{i,j}$ fail with probability threshold $p_{t-1}$.
\end{claimproof}

Next, we prove bounds for $\Pr[Q(X_t,Y_t)\ne\bot]$ in the good and bad cases:

\begin{claim}
    \label{clm:semantics_check}
    The following items hold:
    \begin{itemize}
        \item \textbf{Bad case:} $\Pr[Q(X_t,Y_t)\ne\bot \mid \lnot E_t]\le 2^{-2K}$.
        \item \textbf{Good case:} $\Pr[Q(X_t,Y_t)\ne\bot \mid E_{t-1}]\ge 1-2^{-2K}$.
    \end{itemize}
\end{claim}

\begin{claimproof}
    If $E_t$ does not happen, then there exists $1\le j\le i<t$ such that when running $Q_{i,j}$ on $A$, there exists a (unique) string $s$ outputted with probability $>p_t$, and $s$ is a valid solution to $(X_i, Y_i)$.
    
    Using a straightforward Chernoff bound, when running Lines~\ref{line: check start}-\ref{line: check end} in \autoref{pro:postbpp_xor_ms}, with probability $1-2^{-2K}$, the empirical probability $p_{\sf emp}$ of outputting $s$ will be greater than $p_t\cdot(1-0.01/t^2)>(p_t+p_{t-1})/2$. This means that $Q$ outputs $\bot$ with probability $\ge 1-2^{-2K}$.

    The second bullet can be proved similarly.
\end{claimproof}

Finally, in the good and intermediate case, we show that $S$ holds with high probability.

\begin{claim}
    \label{clm:correctness_bound}
    In \autoref{pro:postbpp_xor_ms}, for a truncated oracle $A$ at level $t-1$, if $E_{t}$ holds for $A$, we have that
    \[\Pr[S\mid A\text{ is sampled}]\ge 1-2^{-\Omega((n_t)^2)}.\]
\end{claim}

\begin{claimproof}
    By our assumption on $t$, for any possible $A$, at least one of the protocols $Q_{i,j}$ ($1\le j\le i\le t$) does not fail with probability threshold $p_t$. If $E_{t}$ happens, then all the protocols $Q_{i,j}$ with $i<t$ (i.e., those that operate on smaller input lengths) fail on $A$ with probability threshold $p_t$. Therefore, at least one of the protocols $Q_{t,j}$ succeeds on $A$. In this case, when we simulate $Q_{t,j}$ in $Q$, with probability $1-2^{-\Omega(n_t^2)}$, the empirical majority answer $s_j$ is indeed the majority answer of $Q_{t, j}$, which is a valid solution for $(X_t,Y_t)$.
\end{claimproof}

Combining the previous results proves that $Q$ succeeds with high probability.

\begin{claim}
    When running the combined protocol $Q$ on any input $(X_t,Y_t)$, 
    \begin{align*}
        \Pr[S\mid Q(X_t,Y_t)\ne \bot]\ge 1-2^{-\Omega(n_t^2)}.
    \end{align*}
\end{claim}

\begin{claimproof}
    Let $G=E_{t-1}\lor\lnot E_{t}$ (note that $E_{t-1}$ and $\lnot E_{t}$ are disjoint). We have that
    \begin{align}
        \label{equ:split_cases}
        &\Pr[S\mid Q(X_t,Y_t)\ne \bot]\nonumber \\
        \ge{}&\min \mleft\{\Pr[S\mid Q(X_t,Y_t)\ne \bot\land G],\Pr[S\mid Q(X_t,Y_t)\ne \bot\land \lnot G]\mright\}.
    \end{align}

    The second term is at least
    \begin{align*}
        \min_{A:G\text{ does not hold for the oracle } A}\Pr[S\mid Q(X_t,Y_t)\ne \bot\land A\text{ is sampled}],
    \end{align*}
    which is at least $1-2^{-\Omega((n_t)^2)}$ by \autoref{clm:correctness_bound} (note that $\lnot G=\lnot E_{t-1}\land E_{t}$, which implies that $E_{t}$ holds for $A$).

    For the first term, we have that
    \begin{align*}
        &\Pr[S\mid Q(X_t,Y_t)\ne \bot\land G] \\
        \ge{}&\frac{\Pr[S\land Q(X_t,Y_t)\ne \bot\mid G]}{ \Pr[Q(X_t,Y_t)\ne \bot\mid G]} \\
        \ge{}& \frac{\Pr[S\land Q(X_t,Y_t)\ne \bot\land G]}{\Pr[Q(X_t,Y_t)\ne \bot\mid E_{t-1}]\cdot \Pr[E_{t-1}\mid G]+\Pr[Q(X_t,Y_t)\ne \bot\mid \lnot E_{t}]\cdot \Pr[\lnot E_{t}\mid G]} \\
        \ge{}& \frac{\Pr[S\land Q(X_t,Y_t)\ne \bot\land E_{t-1}]}{\Pr[Q(X_t,Y_t)\ne \bot\mid E_{t-1}]\cdot \Pr[E_{t-1}]+\Pr[Q(X_t,Y_t)\ne \bot\mid \lnot E_{t}]}\tag{$E_{t-1}$ is contained in $G$} \\
        \ge{}& \frac{\Pr[S\mid E_{t-1}]\cdot \Pr[ Q(X_t,Y_t)\ne\bot\mid E_{t-1}]\cdot \Pr[E_{t-1}]}{\Pr[Q(X_t,Y_t)\ne \bot\mid E_{t-1}]\cdot \Pr[E_{t-1}]+\Pr[Q(X_t,Y_t)\ne \bot\mid \lnot E_{t}]},
    \end{align*}
    where the last line uses the fact that $S$ and $Q(X_t,Y_t)\ne \bot$ are independent.
    
    Using \autoref{clm:semantics_check} and \autoref{clm:correctness_bound}, we have that
    \begin{align*}
        \Pr[S\mid Q(X_t,Y_t)\ne \bot \land G]&\ge \frac{(1-2^{-\Omega(n_t^2)})\cdot (1-2^{-2K})\cdot \Pr[E_{t-1}]}{\Pr[E_{t-1}]+2^{-2K}}.
    \end{align*}
    
    Since $\Pr[E_{t-1}]\ge 2^{-K}$ by \autoref{clm:at_least_one}, we have that
    \begin{align*}
        \Pr[S\mid Q(X_t,Y_t)\ne \bot\mid G]&\ge \frac{(1-2^{-\Omega(n_t^2)})\cdot (1-2^{-2K})\cdot 2^{-K}}{2^{-K} + 2^{-2K}}\ge 1-2^{-\Omega(n_t^2)}.
    \end{align*}
    Therefore, both terms in \eqref{equ:split_cases} are at least $1-2^{-\Omega(n_t^2)}$.
\end{claimproof}

\paragraph{Putting everything together.} Finally, we prove \autoref{lemma: finite barrier}:

\begin{proof}[Proof of \autoref{lemma: finite barrier}]
    The base case of $t=1$ follows from \autoref{cor: LB for pseudodet PostBPP}.

    For $t>1$, assume that the lemma holds for $t-1$ but not for $t$. This implies that \autoref{pro:postbpp_xor_ms} is a $\PostBPP$ communication protocol $Q$ that:
    \begin{itemize}
        \item takes an input of $\XorMissingString(n_t,m_t)$, and outputs $t$ $n_t$-bit strings,
        \item has complexity $\le O(n_t^{46})=o(m_t)$, and
        \item succeeds with probability $\ge 1-2^{-\Theta((n_t)^2)}\ge 1-2^{-\omega(n_t\cdot t)}$ on any input.
    \end{itemize}

    Such a protocol is impossible by \autoref{lem:xor_ms_lb_multi}.
\end{proof}

%% file: ma_half_exponential.tex
\section{Almost-Everywhere Algebrization Barrier for Robust \texorpdfstring{$\MA_\E$}{MAE}}

\label{sec:ma_barrier}

\subsection{Definitions}

\begin{definition}
    \label{def:Socratic_alg}
    An \defn{$\MA\cap\co\MA$ algorithm} $M$ is associated with two verifiers $V_0,V_1$, which are randomized algorithms that take an input $x$ and a proof $\pi$ from Merlin. We say that 
    \begin{itemize}
        \item A verifier $V_k$ \defn{accepts} $x$, if there exists a proof $\pi$ such that $\Pr[V_k(x,\pi)=1]\ge 2/3$.
        \item A verifier $V_k$ \defn{rejects} $x$, if for any proof $\pi$, we have that $\Pr[V_k(x,\pi)=1]\le 1/3$.
    \end{itemize}

    We define the notion of \defn{robust} algorithms. Intuitively, an $\MA$ algorithm is robust if its truth table is never easy for $A$-oracle circuits. We do however allow an algorithm to be semantically incorrect on some oracle, which is not considered as an error.

    Formally, we say that $M$ is \defn{robust w.r.t.~$h(n)$ on length $n$} if, for \emph{any} oracle $A$ and any $2^n$-bit string $w$ that has $h(n)$ size $A$-oracle circuits, there exists at least one input $x$ of length $n$, such that $V_{w_x}^{\widetilde{A}}$ rejects $x$. That is, Merlin cannot convince Arthur to output the truth table $w$.
\end{definition}

The following definitions characterize $\FMA$ communication protocols. Such protocols are able to compute the truth tables of $\MA\cap\co\MA$ algorithms.

\begin{definition}
    An \defn{$\FMA$ communication protocol} $P$ for a search problem $f$ (over domain $\mathcal{X\times Y}$ and range $\mathcal{O}$) is defined as a set of randomized protocols $P_{w,\pi}$, for each pair of answer $w\in \mathcal{O}$ and proof $\pi$ from Merlin. Let $k$ denote the length of the proof. That is, $\pi\in \{0,1\}^k$.

    We say that $P$ \defn{solves $(X,Y)$}, if
    \begin{itemize}
        \item For some $w,\pi$ where $w$ is a solution to $(X,Y)$, $\Pr[P_{w,\pi}(X,Y)=1]\ge 2/3$.
        \item For any $w,\pi$ where $w$ is \emph{not} a solution to $(X,Y)$, $\Pr[P_{w,\pi}(X,Y)=1]\le 1/3$.
    \end{itemize}

    The \defn{complexity} of $P$ is defined as $k+C$, where $C$ is the maximum complexity of any randomized protocol $P_{w,\pi}$. Here, the complexity of a randomized protocol is defined as the maximum number of bits communicated on any input, plus the number of random bits used by the protocol. Note that we only allow public randomness (which is available to both Alice and Bob but not Merlin). %

    A protocol $P$ is \defn{robust}, if for any $(X,Y)$ and for any $w,\pi$ where $w$ is \emph{not} a solution to $(X,Y)$, $\Pr[P_{w,\pi}(X,Y)=1]\le 1/3$.
\end{definition}

\subsection{\texorpdfstring{$\XorMissingString$}{XOR-Missing-String} Lower Bounds Against Robust Protocols}

The main engine of our barrier is the following lemma, which shows that robust protocols are much weaker than general protocols.

\ZeroErrorLB*
\begin{proof}
    We only prove the lemma for the case where $C<m/2$, since otherwise the bound is straightforward.
    
    As an $\FMA$ communication protocol, $P$ can be viewed as a collection of at most $2^n\cdot 2^C$ randomized protocols, where each protocol $P_{w,\pi}$ corresponds to a pair $(w,\pi)$ of candidate answer and proof, and $P_{w,\pi}$ outputs $1$ if and only if Arthur accepts Merlin's proof $\pi$, which proves that $w$ is a solution.
    
    Let $P'_{w,\pi}$ denote an error-reduced version of $P_{w,\pi}$. Specifically, $P'_{w,\pi}$ runs $P_{w,\pi}$ repeatedly for $\Theta(m/C)$ times, each time using independent randomness, and then outputs the majority answer. We choose the number of repetitions carefully, so that
    \begin{enumerate}
        \item $P'_{w,\pi}$ has complexity $<m/2$.
        \item On input $(X,Y)$, if $w$ is not a solution, then
        \begin{align}
            \label{equ:socratic}
            \Pr[P'_{w,\pi}(X,Y)=1]\le 2^{-\Omega(m/C)}.
        \end{align}
    \item \label{item: amplify correctness} If $P$ solves $\XorMissingString$ on $(X,Y)$, then for some $(w,\pi)$ where $w$ is a solution to $(X,Y)$, we have that
    \begin{align*}
        \Pr[P'_{w,\pi}(X,Y)=1]\ge 1-2^{-\Omega(m/C)}.
    \end{align*}
    \end{enumerate}

    \autoref{item: amplify correctness} means that the fraction of inputs solvable by $P$ is at most
    \begin{align*}
        (1+2^{-\Theta(m/C)})\cdot \sum_{w,\pi}\Pr_{(X,Y)}[P'_{w,\pi}(X,Y)=1],
    \end{align*}
    where the probability is over the choice of $(X,Y)$ and the internal randomness of $P'_{w,\pi}$. To bound this value, we think of $P'_{w,\pi}$ as a distribution of deterministic protocols. For each deterministic protocol, we have the following claim, which is a relation between the probabilities of true and false positives.

    \begin{claim}
        \label{clm:det_socratic}
        Let $w$ be an $n$-bit string, and let $P$ be a deterministic protocol of complexity $C$, which outputs $0$ or $1$. Let
        \begin{align*}
            p&=\Pr_{(X,Y)}[P(X,Y)=1], \\
            p_{\text{wrong}}&=\Pr_{(X,Y)}[P(X,Y)=1\land w\text{ does not solve }(X,Y)].
        \end{align*}
        We have that
        \begin{align*}
            p\le 2^{-m+2}\cdot 2^C+p_{\text{wrong}}\cdot 2^{2n+2}.
        \end{align*}
    \end{claim}

    \begin{claimproof}
        $P$ can be characterized as a collection of at most $2^C$ rectangles that are pairwise disjoint and cover the entire input space, such that the output of $P$ is the same for inputs that belong to the same rectangle. We say that a rectangle is \emph{large} if it has size $\ge 2^{2nm-m+2}$, and say that a rectangle is \emph{small} otherwise.

        We fix one large rectangle, and consider its contribution to $p$ and $p_{\text{wrong}}$. Using \autoref{lem:rect_error_prob_lb}, when a rectangle $R$ is large, there must exist a $2^{-(2n+2)}$ fraction of input in $R$, on which $w$ is not a solution. Therefore, if $P$ outputs $1$ on $R$, then $R$ contributes $|R|/2^{2nm}$ to $p$, and contributes at least $|R|\cdot 2^{-(2n+2)}/2^{2nm}$ to $p_{\text{wrong}}$. We therefore have the following relation between $p$ and $p_{\text{wrong}}$:
        \begin{align*}
            p\le 2^{-m+2}\cdot 2^C+p_{\text{wrong}}\cdot 2^{2n+2},
        \end{align*}
        where the first term upper bounds the total size of small rectangles.
    \end{claimproof}

    For each $P'_{w,\pi}$, by applying \autoref{clm:det_socratic} on each deterministic protocol in its distribution, we have that
    \begin{align}
        \label{equ:relate}
        \Pr_{(X,Y)}[P'_{w,\pi}(X,Y)=1]\le 2^{-m+2}\cdot 2^C+\Pr_{(X,Y)}[P'_{w,\pi}(X,Y)=1\land w\text{ does not solve }(X,Y)]\cdot 2^{2n+2}.
    \end{align}
    By \eqref{equ:socratic}, whenever $w$ does not solve $(X,Y)$, we have that $\Pr[P'_{w,\pi}(X,Y)=1]\le 2^{-\Omega(m/C)}$, where the probability is over the internal randomness of $P'_{w,\pi}$. Summing over all $(X,Y)$, we have that
    \begin{align*}
        \Pr_{(X,Y)}[P'_{w,\pi}(X,Y)=1\land w\text{ does not solve }(X,Y)]\le 2^{-\Omega(m/C)}.
    \end{align*}
    Plugging this into \eqref{equ:relate}, we thus have that
    \begin{align*}
        \Pr_{(X,Y)}[P'_{w,\pi}(X,Y)=1]\le 2^{-\Omega(m/C)+O(C+n)}.
    \end{align*}
    Summing this over for every $P'_{w,\pi}$, we have that (note that the number of pairs $w,\pi$ is at most $2^{O(C+n)}$)
    \begin{align*}
        (1+2^{-\Theta(m/C)})&\cdot \sum_{w,\pi}\Pr_{(X,Y)}[P'_{w,\pi}(X,Y)=1]\le 2^{-\Omega(m/C)+O(C+n)}.\qedhere
    \end{align*}
\end{proof}

\subsection{The Barrier}

We say that a function $h(n)$ is \defn{super-half-exponential}, if $h$ is strictly monotonically increasing, $h(n) = \omega(n)$, and $h(h(n)/n) \ge 2^n$. In the rest of this section, let $h(n)$ be any fixed super-half-exponential function.

Let $\RobMAE$ denote the class of languages decidable by algorithms that are in $\MA_\E \cap \co\MA_\E$ and are robust w.r.t.~$h(n)$ on every input length. We present the following barrier result.

\begin{theorem}
\label{thm:ma_ae_barrier}
There exists an oracle $A$ such that
$     (\E^{\widetilde{A}} \cup (\RobMAE)^{\widetilde{A}}) \subseteq \SIZE^A[h(n)],
    $
where $\widetilde{A}$ denotes the multilinear extension of $A$.
\end{theorem}

\subsubsection{Oracle Structure} In this proof, we say that an oracle $A$ is \defn{well-structured}, if it encodes one instance of $\XorMissingString$ for each input length $i$. Formally, for every $i\in \mathbb N$, the truth table of $A$ on inputs of length $h(i)$ encodes an instance $(X_i,Y_i)$ of $\XorMissingString(n_i=2^i,m_i=2^{h(i)/10})$ (note that this instance can be described in $2n_im_i\le 2^{h(i)}$ bits), where $X_i$ (resp.~$Y_i$) is encoded in the truth table for $A$ on inputs that start with $0$ (resp.~$1$). When querying $A$ on an input of length other than $h(i)$ for some $i$, the oracle returns $0$.

These parameters are chosen such that, if an algorithm $M^{\widetilde{A}}$ is semantically correct and does not solve $(X_i,Y_i)$ (i.e., its truth table on inputs of length $=i$ is not a solution), then there exists an $A$-oracle circuit of size less than $h(i)$ that computes the truth table of $M^{\widetilde{A}}$ on inputs of length $i$. Indeed, the truth table must be of the form $x\oplus y$, where $x\in X_i$ and $y\in Y_i$. Such strings can be computed by $A$-oracle circuits of size $2\log m_i+O(i)<h(i)$, by hardcoding the index of $x$ and $y$ in $X_i$ and $Y_i$ respectively.

\subsubsection{Converting Algorithms to Protocols} Similar to the previous barriers, we start by enumerating the algorithms in $\E$ and in $\RobMAE$. The two types of algorithms are enumerated separately. For $\E$, let $M_1,M_2,\ldots$ be an enumeration of deterministic algorithms in $\DTIME[2^n]$. Let $P_{i,j}$ be a deterministic protocol that outputs the truth table of $M_j^{\widetilde{A}}$ on inputs of length $i$. Note that $P_{i,j}$ runs in time $2^{O(i)}$. We aim to ensure that, for any sufficiently large $i$ and $1\le j\le i$, $P_{i,j}$ does not solve $(X_i,Y_i)$.

For $\RobMAE$, let $(V'_{j,0},V'_{j,1})_{j\in \mathbb N}$ be a syntactic enumeration of pairs of verifiers, where each verifier is a randomized algorithm that receives an input $x$ and a proof $\pi$ and runs in time $2^{|x|}$ (in particular, the length of $\pi$ is at most $2^{|x|}$). Let $M'_j$ be the $\MA\cap\co\MA$ algorithm associated with $(V'_{j,0},V'_{j,1})$. For $i\in \mathbb N$ and $1\le j\le i$, let $Q_{i,j}$ be the $\FMA$ protocol of complexity $2^{O(i)}$ that outputs the truth table of $(M'_j)^{\widetilde{A}}$ on inputs of length $i$, in the case that the algorithm is robust. More formally, $Q_{i,j}$ is defined as follows.
\begin{itemize}
    \item If $(M'_j)^{(-)}$ is not robust w.r.t.~$h(i)$ on inputs of length $i$, then $Q_{i,j}$ rejects every pair $(w,\pi)$ of candidate answer and proof (i.e., it is the trivial robust protocol).
    \item Otherwise, $Q_{i,j}$ simulates the execution of $(M'_j)^{\widetilde{A}}$ on inputs of length $i$ (repeatedly for $2^{O(i)}$ times on each input) and outputs its truth table. 
    
    More specifically, $Q_{i,j}$ corresponds to the following set of verifiers $\{V_{w,\pi}\}$: $V_{w,\pi}$ interprets $\pi$ as a concatenation of proofs $\{\pi_x:x\in\{0,1\}^i\}$, one for each input of length $i$. For each input $x$ of length $i$, $V_{w,\pi}$ checks whether the $x$-th bit of $w$ (denoted as $w_x$) is the output of $(M'_j)^{\widetilde{A}}(x)$, by simulating the execution of $(M'_j)^{\widetilde{A}}(x)$ for $2^{i}$ times.  In each simulation, $V_{w,\pi}$ checks whether $(M'_j)^{\widetilde{A}}(x)=w_x$ when given the proof $\pi_x$. If for every $x$, the majority of simulations accept, then $V_{w,\pi}$ outputs $1$; otherwise, it outputs $0$.

    Each $V_{w,\pi}$ has complexity $2^{O(i)}$, and the length of $\pi$ is also $2^{O(i)}$. Therefore, the complexity of $Q_{i,j}$ is $2^{O(i)}$. Moreover, $Q_{i,j}$ is a robust  protocol for $(X_i,Y_i)$. This is because $M'_j$ is robust, which means that for any non-solution $w$ of $(X_i,Y_i)$, there exists some $x$ such that $(V'_{j,s_x})^{\widetilde{A}}$ rejects $x$. Therefore, when $V_{w,\pi}$ checks $x$, the majority of simulations will reject.
\end{itemize}

In the following construction, we aim to ensure that, for any sufficiently large $i$ and $1\le j\le i$, $Q_{i,j}$ does not solve $(X_i,Y_i)$. If this holds, then for any algorithm $M'_j$ in $\RobMAE$, there exists a family of $A$-oracle circuits of size $h(i)$ that computes the truth table of $(M'_j)^{\widetilde{A}}$.

Note that for every $i\in \mathbb N$ and $1\le j\le i$, since $P_{i,j}$ and $Q_{i,j}$ both simulate algorithms that run in time $2^{i}$, both protocols only access the oracle $A$ on inputs of length $\le 2^i$. In other words, when $2^i<h(i')$, these protocols cannot access $(X_{i'},Y_{i'})$.%

\subsubsection{The Inductive Guarantee} 

We construct the oracle $A$ inductively, where at step $i$, we maintain a rectangle $R_i$ of well-structured oracles (i.e., $R_i=\mathcal{A}_i\times \mathcal{B}_i$ where $\mathcal{A}_i$ is a subset of Alice's inputs and $\mathcal{B}_i$ is a subset of Bob's inputs). The goal is to refute all protocols by shrinking the rectangle, where we hope to guarantee that for any $i$, the protocols $P_{\le i,j},Q_{\le i,j}$ fail on any oracle in $R_i$ (in reality, we achieve a weaker guarantee; see below).

The following constant parameters are used in the construction. The big-O notations will not hide any dependence on these constants.

\begin{itemize}
    \item Let $c_1\ge 1$ be a constant such that any protocol $P_{i,j}$ or $Q_{i,j}$ have complexity at most $2^{c_1i}=n_i^{c_1}$.
    \item Let $c_2\ge 1$ be a constant such that, for any $i\ge i_0$ and any robust protocol $Q$ of complexity $n_i^{c_1}$ attempting to solve $\XorMissingString(n_i,m_i)$, $Q$ solves at most $2^{-c_2m_i/n_i^{c_1}}$ fraction of all inputs. Such a constant exists by \autoref{lem:socratic_lb}. 
    \item Let $c_3,c_4$ be constants such that $c_1\ll c_3\ll c_4$.
    \item Let $i_0$ be a sufficiently large constant that depends on $c_1$ through $c_4$. We only refute the protocols $P_{i,j},Q_{i,j}$ where $i\ge i_0$.
\end{itemize}

We ensure the following properties for $R_i$. Intuitively, $R_1\supset R_2\supset \ldots$ is a sequence of rectangles representing increasingly stronger restrictions on the oracle $A$, where $R_i$ (ideally) refutes all protocols $P_{i',j},Q_{i',j}$ where $i_0\le i'\le i$. 
\begin{enumerate}
    \item \textbf{Containment:} For any $i>1$, $R_i$ is a subrectangle of $R_{i-1}$. 
    \item \textbf{Measurability:}\label{property: measurability} To decide whether an oracle $A$ is in $R_i$, we only need to look at the first $h(i)$ instances $(X_{\le h(i)},Y_{\le h(i)})$ of $A$. 
    
    Although $R_i$ is a set of uncountably many oracles, measurability implies that we can treat it as a set of finite oracles (up to length $h(i)$). Hence, we can reasonably talk about the ``size'' of $R_i$, as well as the ``uniform distribution'' over $R_i$.
    \item \textbf{Common prefix:} \label{property: common prefix} All the oracles in $R_i$ agree on the first $i$ instances $(X_{\le i},Y_{\le i})$. We denote this common prefix as $(X^{\text{fixed}}_1,Y^{\text{fixed}}_1),\ldots, (X^{\text{fixed}}_i,Y^{\text{fixed}}_i)$. Since for each $i\ge 1$, $R_{i+1}$ is a subrectangle of $R_i$, we have that $(X^{\text{fixed}}_1, Y^{\text{fixed}}_1), \dots, (X^{\text{fixed}}_i, Y^{\text{fixed}}_i)$ is also a prefix of every $R_{i'}$ ($i' > i$).
    \item \textbf{Largeness:} \label{property: largeness} $R_i$ contains at least a $2^{-n_i^{c_3}}$ fraction of oracles that agree on $(X^{\text{fixed}}_{\le i},Y^{\text{fixed}}_{\le i})$. Formally, for a random well-structured oracle $A$,
    \begin{align*}
        \Pr_A[A\in R_i\mid A\text{ agrees with }(X^{\text{fixed}}_{\le i},Y^{\text{fixed}}_{\le i})]\ge 2^{-n_i^{c_3}}.
    \end{align*}
    \item \textbf{Refuting deterministic protocols:} \label{property: refuting deterministic protocols} For $i_0\le i'\le i$, $1\le j\le i'$ and any oracle $A\in R_i$, $P_{i',j}$ does not solve $(X_{i'},Y_{i'})$ on $A$.
    \item \textbf{Refuting robust protocols:} \label{property: refuting Socratic protocols}
    For $i_0\le i'\le i$, $1\le j\le i'$:
    \begin{enumerate}
        \item \textbf{Small protocols:} \label{property 6: small} If $2^{i'}\le h(i)$, then for any oracle $A\in R_i$, $Q_{i',j}$ does not solve $(X_{i'},Y_{i'})$ on $A$. (In this case, $Q_{i', j}$ only has access to the first $i$ instances in $A$, which are equal to $\{(X^{\text{fixed}}_{i''}, Y^{\text{fixed}}_{i''})\}_{i'' \le i}$.)
        \item \textbf{Large protocols:}\label{property 6: large} If $2^{i'}> h(i)$, then
        \begin{align*}
            \Pr_A[Q_{i',j}\text{ solves }(X_{i'},Y_{i'})\text{ on }A\mid A\in R_i]\le 2^{-c_2m_{i'}/n_{i'}^{c_1}+n_i^{c_4}}.
        \end{align*}
        We provide some intuition about this requirement, in particular about the right-hand side $2^{-c_2m_{i'}/n_{i'}^{c_1} + n_i^{c_4}}$. The first term $c_2m_{i'}/n_{i'}^{c_1}$ is the exponent in the upper bound of \autoref{lem:socratic_lb}, which is the probability that $Q_{i',j}$ solves $(X_{i'},Y_{i'})$ on a random oracle. Since $R_i$ is only a subset of all oracles, the probability of $Q_{i',j}$ solving a random oracle in $R_i$ might be larger, and this probability could increase as we reduce the size of $R_i$. This behavior is captured by the second term $n_i^{c_4}$ in the exponent.

    \end{enumerate}
\end{enumerate}

\paragraph{Why half-exponential?} The probability bound of Property \ref{property 6: large} is the reason that our proof requires $h(n)$ to be super-half-exponential. In order for (\ref{property 6: large}) to be nontrivial, we require that $2^{-c_2m_{i'}/n_{i'}^{c_1}+n_i^{c_4}}<1$ for any $i',i\ge i_0$ such that $2^{i'}>h(i)$. Recall that $m_{i}=2^{h(i)/10}$ and $n_i=2^i$, this requirement loosely translates to $h(i')\ge i$ whenever $2^{i'}>h(i)$ (up to some polynomial). Assuming that $h$ is monotone, $h(i')\ge i$ is equivalent to $h(h(i'))\ge h(i)$, and our requirement becomes $h(h(i'))>2^{i'}$ for any $i'$, i.e., $h$ must be super-half-exponential.

Formally, given that $h$ is super-half-exponential, we can show that:

\begin{restatable}{claim}{halfExpMath}
    \label{clm:half_exp_math}
    For $i_0\le i'\le i$ and $1\le j\le i'$ such that $h(i-1)<2^{i'}\le h(i)$, we have that $2^{-c_2m_{i'}/n_{i'}^{c_1}+n_i^{c_4}}<1$.
\end{restatable}

\begin{proof}
    The above inequality is equivalent to
    \begin{align*}
        m_{i'}> n_{i'}^{c_1}\cdot n_i^{c_4}/c_2.
    \end{align*}
    Recall that $n_{i'}=2^{i'}$ and $m_{i'}=2^{h(i')/10}$. Taking logarithm on both sides, this is equivalent to showing that
    \begin{align*}
        h(i')/10&> c_1\cdot i'+c_4\cdot i+O(1)=c_4\cdot i\cdot (1+o(1)).
    \end{align*}
    Suppose that this does not hold. Then we have
    \begin{align*}
        h(i')/O(c_4)&\le i-1.
    \end{align*}
    Since $h$ is monotone, we have that
    \begin{align*}
        h(h(i')/O(c_4))&\le h(i-1).
    \end{align*}
    However, since $h$ is super-half-exponential, we have that $h(h(i')/i')\ge 2^{i'}$. This contradicts $2^{i'}> h(i-1)$.
\end{proof}

If we can construct such a sequence of rectangles, then the oracle $A^{\text{fixed}}=(X^{\text{fixed}}_{i},Y^{\text{fixed}}_{i})$ (formally, $A^{\text{fixed}}(x)=1$ if and only if $A'(x)=1$ for every $A'\in R_n$, where $n$ is sufficiently large) proves the theorem, because it lies in every rectangle $R_n$, and thus refutes every protocol.

\subsubsection{The Construction} When $i<i_0$, we let $(X^{\text{fixed}}_{\le i},Y^{\text{fixed}}_{\le i})$ be an arbitrary prefix, and let $R_i$ be the set of all oracles that agree with $(X^{\text{fixed}}_{\le i},Y^{\text{fixed}}_{\le i})$. %

When $i\ge i_0$, suppose that we have constructed $R_{i-1}$. To construct $R_i$, we start by letting $R^{0}_i\gets R_{i-1}$, then gradually shrink $R^0_i$ to ensure that the above properties are satisfied.

\paragraph{Step 1: Refuting the deterministic protocols.} We start by finding a large subrectangle in $R^0_i$ on which all deterministic protocols fail. Note that we only have to consider protocols $P_{i,j}$ for each $1\le j\le i$, since the other protocols $P_{<i,j}$ already fail on every oracle in $R_{i-1}$.

For each $1\le j\le i$, since the deterministic protocol $P_{i,j}$ runs in time $n_i^{c_1}$, there exists a subrectangle $R'$ of $R^0_i$, such that $R'$ contains at least $2^{-n_i^{c_1}}$ fraction of $R^0_i$, and $P_{i,j}$ outputs the same string for any oracle in $R'$.

Next, we use \autoref{cor:rect_error_prob_lb_aux} to find a subrectangle $R''$ of $R'$, such that $R''$ contains at least $2^{-O(n_i)}$ fraction of $R'$, and that for any oracle $A\in R''$, $P_{i,j}$ does not solve $(X_i,Y_i)$ on $A$.

\RectAux*

Note that \autoref{cor:rect_error_prob_lb_aux} is not na\"ively applicable to $R'$, because the oracles in $R'$ are infinitely long. However, by the property of measurability (\ref{property: measurability}), we may pretend that the oracles in $R'$ only encode finitely many instances $(X_{*},Y_{*})$. It is then possible to interpret the instances $(X_{\ne i},Y_{\ne i})$ as the auxiliary input and apply \autoref{cor:rect_error_prob_lb_aux}. Finally, we update $R^0_i\leftarrow R''$ to refute $P_{i,j}$. 

After considering all $1\le j\le i$, we have that, for any $1\le j\le i$ and any oracle $A\in R^0_i$, $P_{i,j}$ does not solve $(X_i,Y_i)$ on $A$. Let $R^1_i$ denote the resulting rectangle after refuting all $P_{i,j}$. Note that the size of $R^1_i$ is at least a $2^{-i\cdot O(n_i^{c_1})} \ge 2^{-O(n_i^{c_1})}$ fraction of the size of $R_{i-1}$. After this, the current rectangle $R^1_i$ satisfies all properties except \ref{property: common prefix} and \ref{property: refuting Socratic protocols}:

\begin{itemize}
    \item Containment is satisfied because we only shrink the rectangle.
    \item Measurability is satisfied because it was satisfied by $R_{i-1}$, and we only look at the first $h(i)$ instances of $A$ to decide the answer of each $P_{i,j}$. Also, when applying \autoref{cor:rect_error_prob_lb_aux}, it suffices to check the value of $(X_i,Y_i)$. To see this, recall that in \autoref{cor:rect_error_prob_lb_aux}, the subrectangle is constructed by fixing some $n_i$-bit string $s$, and restricting the rectangle to those oracles for which both $X_i$ and $Y_i$ contain $s$.
    \item $R^1_i$ is still large because we only shrink it by a factor of $2^{-n_i^{O(c_1)}}$:
    \begin{align}
        \label{equ:largeness_temp}
        \Pr_A[A\in R^1_i\mid A\text{ agrees with }(X^{\text{fixed}}_{\le i-1},Y^{\text{fixed}}_{\le i-1})]\ge 2^{-n_{i-1}^{c_3}-n_i^{O(c_1)}}=2^{-O(n_{i-1}^{c_3})}.
    \end{align}
    Note that this is not exactly the same as Largeness (Property \ref{property: largeness}) because we have not yet fixed $(X^{\text{fixed}}_i,Y^{\text{fixed}}_i)$.
    \item $R^1_i$ refutes all deterministic protocols because $R_{i-1}$ already refutes all $P_{i',j}$ where $i_0\le i'\le i-1$, and we have just ensured that $R^1_i$ refutes all $P_{i,j}$.
\end{itemize}

As for the robust protocols, we now have the following guarantee: For $i_0\le i'\le i$, $1\le j\le i'$,
\begin{enumerate}[label=(\alph*)]
    \item \textbf{Small protocols:} If $2^{i'}\le h(i-1)$, then for any oracle $A\in R^1_i$, $Q_{i',j}$ does not solve $(X_{i'},Y_{i'})$ on $A$. This holds by transitivity from $R_{i-1}$.
    \item \textbf{Large protocols:} If $2^{i'}> h(i-1)$, then
    \begin{align}
        \label{equ:large_socratic_temp}
        \Pr_A[Q_{i',j}\text{ solves }(X_{i'},Y_{i'})\text{ on }A\mid A\in R^1_i]\le 2^{-c_2m_{i'}/n_{i'}^{c_1}+O(n_{i-1}^{c_4})}.
    \end{align}
    
    For $i'<i$, this holds because $R_{i-1}$ satisfies Property \ref{property: refuting Socratic protocols} on $i-1$, i.e., it refutes the protocol $Q_{i',j}$ if $i'<i$. Since we only shrink $R^1_i$ by a $2^{-n_i^{O(c_1)}}$ factor, the success probability of $Q_{i',j}$ only increases by a $2^{n_i^{O(c_1)}}=2^{O(n_{i-1}^{c_4})}$ factor.
    
    For $i'=i$, our bound comes from applying \autoref{lem:socratic_lb} to the set of all oracles (that agree with $(X^{\text{fixed}}_{\le i-1},Y^{\text{fixed}}_{\le i-1})$). Note that
    \begin{align}
        \label{equ:large_socratic_temp2}
        \Pr_A[Q_{i,j}\text{ solves }(X_i,Y_i)\text{ on }A\mid A\text{ agrees with }(X^{\text{fixed}}_{\le i-1},Y^{\text{fixed}}_{\le i-1})]\le 2^{-c_2m_i/n_i^{c_1}},
    \end{align}
    which is a corollary of \autoref{lem:socratic_lb}. Indeed, if \eqref{equ:large_socratic_temp2} does not hold, then we can fix an infinite sequence $(X^{\text{bad}}_{>i},Y^{\text{bad}}_{>i})$ of instances, such that
    \begin{align*}
        \Pr_A[Q_{i,j}\text{ solves }(X_i,Y_i)\text{ on }A\mid A\text{ agrees with }(X^{\text{fixed}}_{\le i-1},Y^{\text{fixed}}_{\le i-1}) \text{ and }(X^{\text{bad}}_{>i},Y^{\text{bad}}_{>i})]> 2^{-c_2m_i/n_i^{c_1}}.
    \end{align*}
    Using this, we can construct a robust protocol $Q$ that breaks the bound of \autoref{lem:socratic_lb}: On input $(X_i,Y_i)$, $Q$ combines it with $(X^{\text{fixed}}_{\le i-1},Y^{\text{fixed}}_{\le i-1})$ and $(X^{\text{bad}}_{>i},Y^{\text{bad}}_{>i})$ to form a well-structured oracle $A$, then simulates $Q_{i,j}$ on $A$. The success probability of $Q$ is exactly the above probability, which contradicts \autoref{lem:socratic_lb}.

    Combining \eqref{equ:large_socratic_temp2} and \eqref{equ:largeness_temp}, we have
    \begin{align*}
        \Pr_A[Q_{i,j}\text{ solves }(X_i,Y_i)\text{ on }A\mid A\in R^1_i]\le 2^{-c_2m_i/n_i^{c_1}+O(n_{i-1}^{c_3})}
    \end{align*}
    for every $1\le j\le i$.
\end{enumerate}

\paragraph{Step 2: Extending the common prefix.} Next, we find an instance $(X^{\text{fixed}}_i,Y^{\text{fixed}}_i)$, and let $R_i$ be the restriction of $R^1_i$ to the set of oracles that agree with $(X^{\text{fixed}}_{\le i},Y^{\text{fixed}}_{\le i})$. In the following, we show that there exists a choice of $(X^{\text{fixed}}_i,Y^{\text{fixed}}_i)$ such that every property is satisfied. The proof is by showing that a random instance $(X_i,Y_i)$ satisfies every property with nonzero probability.

As the other properties are straightforward, we only show that Properties \ref{property: largeness} and \ref{property: refuting Socratic protocols} are satisfied. More precisely, we show that:

\begin{lemma}
    \label{lem:extend_prefix}
    There exists a choice of $(X^{\text{fixed}}_i,Y^{\text{fixed}}_i)$, such that when $R_i$ is set to be the subset of oracles in $R^1_i$ that agree with $(X^{\text{fixed}}_i,Y^{\text{fixed}}_i)$, we have that
    \begin{itemize}
        \item \textbf{Largeness:}
        \begin{align*}
            \Pr_A[A\in R_i\mid A\text{ agrees with }(X^{\text{fixed}}_{\le i},Y^{\text{fixed}}_{\le i})]\ge 2^{-n_i^{c_3}}.
        \end{align*}
        \item \textbf{Refuting robust protocols:} For $i',j$ such that $i_0\le i'\le i$, $1\le j\le i'$, we have that
        \begin{enumerate}[label=(\alph*)]
            \item \textbf{Small protocols:} If $2^{i'}\le h(i-1)$, then for any oracle $A\in R_i$, $Q_{i',j}$ does not solve $(X_{i'},Y_{i'})$ on $A$.
        \item \textbf{Large protocols:} If $2^{i'}> h(i-1)$, then
        \begin{align*}
            \Pr_A[Q_{i',j}\text{ solves }(X_{i'},Y_{i'})\text{ on }A\mid A\in R_i]\le 2^{-c_2m_{i'}/n_{i'}^{c_1}+n_i^{c_4}}.
        \end{align*}
        \end{enumerate}
    \end{itemize}
\end{lemma}

Note that the property stated for robust protocols is different from Property \ref{property: refuting Socratic protocols}, in that the large protocols are defined to be those that satisfy $2^{i'}>h(i-1)$ instead of $2^{i'}>h(i)$. However, as we will show in \autoref{clm:final_temp}, \autoref{lem:extend_prefix} implies Property \ref{property: refuting Socratic protocols}.

\begin{proof}

    We first show that largeness is satisfied with nontrivial probability, then show that conditioned on largeness being satisfied, each robust protocol is refuted with high probability.

    \paragraph{Largeness.} For an instance $(X_i,Y_i)$, let $p_{(X_i,Y_i)}$ denote
    \begin{align*}
        \Pr_A[A\in R^1_i\mid A\text{ agrees with }(X^{\text{fixed}}_{\le i-1},Y^{\text{fixed}}_{\le i-1})\land A\text{ agrees with }(X_i,Y_i)],
    \end{align*}
    i.e., the fraction of $R^1_i$ among the oracles that agree with $(X^{\text{fixed}}_{\le i-1},Y^{\text{fixed}}_{\le i-1})$ and $(X_i,Y_i)$. In order to satisfy Property \ref{property: largeness}, we must choose an instance $(X_i,Y_i)$ such that $p_{(X_i,Y_i)}\ge 2^{-n_i^{c_3}}$. Let
    \begin{align*}
        p_{\text{large}}=\Pr_{(X_i,Y_i)}[p_{(X_i,Y_i)}\ge 2^{-n_i^{c_3}}].
    \end{align*}
    Intuitively, since $\Pr[A\in R^1_i\mid A\text{ agrees with }(X^{\text{fixed}}_{\le i-1},Y^{\text{fixed}}_{\le i-1})]$ is guaranteed to be large, $p_{(X_i,Y_i)}$ should be large on average. By applying a simple Markov bound, we can show that $p_{(X_i,Y_i)}$ is also large with decent probability:
    \begin{claim}
        \label{clm:large_fraction}
        $p_{\text{large}}\ge 2^{-n_i^{c_3}}$.
    \end{claim}
    \begin{claimproof}
        Using the notation $p_{(X_i,Y_i)}$, \eqref{equ:largeness_temp} directly translates to
        \begin{align*}
            \mathbb{E}_{(X_i,Y_i)}[p_{(X_i,Y_i)}]&\ge 2^{-O(n_{i-1}^{c_3})}.
        \end{align*}
        Since $\Pr_{(X_i,Y_i)}[p_{(X_i,Y_i)}\le 1]=1$ by definition, we have that
        \begin{align*}
            \mathbb{E}_{(X_i,Y_i)}[p_{(X_i,Y_i)}]&\le 
            p_{\text{large}}\cdot 1+(1-p_{\text{large}})\cdot 2^{-n_i^{c_3}} \nonumber \\
            &\le p_{\text{large}}+2^{-n_i^{c_3}}.
        \end{align*}
        The claim holds because $c_3$ is sufficiently large (note that $n_i=2^i$).
    \end{claimproof}
    \paragraph{Refuting robust protocols.} Next, we apply a union bound over the robust protocols, and show that conditioned on $p_{(X_i,Y_i)}\ge 2^{-n_i^{c_3}}$, with high probability, every protocol satisfies the property stated in the lemma.
    
    Let $i_0\le i'\le i$ and $1\le j\le i'$. In the small case (i.e., $2^{i'}\le h(i-1)$), $Q_{i',j}$ does not solve $(X_{i'},Y_{i'})$ on any oracle in $R_i$, and this property is preserved regardless of $(X^{\text{fixed}}_i,Y^{\text{fixed}}_i)$. Therefore, we only need to consider the case where $2^{i'}> h(i-1)$. For any $Q_{i',j}$ and $(X_i,Y_i)$, let $q_{i',j,(X_i,Y_i)}$ denote
    \begin{align*}
        \Pr_A[Q_{i',j}\text{ solves }(X_{i'},Y_{i'})\text{ on }A\mid A\in R^1_i\land A\text{ agrees with }(X_i,Y_i)],
    \end{align*}
    i.e., the fraction of oracles in $R_i$ that $Q_{i',j}$ solves, if $(X^{\text{fixed}}_i,Y^{\text{fixed}}_i)$ is chosen to be $(X_i,Y_i)$. We show that for a random instance $(X_i,Y_i)$, conditioned on $p_{(X_i,Y_i)}\ge 2^{-n_i^{c_3}}$, $q_{i',j,(X_i,Y_i)}$ is small with high probability. Again, the proof holds by a simple Markov bound, where we use the fact that $Q_{i',j}$ only solves a small fraction of oracles in $R^1_i$.
    
    \begin{claim}
        \label{clm:socratic_union}
        Let $i_0\le i'\le i$, $1\le j\le i'$, be such that $2^{i'}> h(i-1)$. We have that
        \begin{align*}
            \Pr_{(X_i,Y_i)}[q_{i',j,(X_i,Y_i)}> 2^{-c_2m_{i'}/n_{i'}^{c_1}+n_i^{c_4}}\mid p_{(X_i,Y_i)}\ge 2^{-n_i^{c_3}}]\le 2^{-i}.
        \end{align*}
    \end{claim}
    \begin{claimproof}
        The proof is by investigating the value $\mathbb{E}_{(X_i,Y_i)}[q_{i',j,(X_i,Y_i)}\cdot p_{(X_i,Y_i)}]$, which is the probability that $Q_{i',j}$ solves a random oracle in $R^1_i$ (conditioned on the common prefix). Indeed, we can show that this is equal to
        \begin{align*}
            \Pr_A[Q_{i',j}\text{ solves }(X_{i'},Y_{i'})\text{ on }A\land A\in R^1_i\mid A\text{ agrees with }(X^{\text{fixed}}_{\le i-1},Y^{\text{fixed}}_{\le i-1})].
        \end{align*}
        \begingroup \allowdisplaybreaks
        Formally, we have that (for simplicity, let $B$ denote the event $[Q_{i',j}\text{ solves }(X_{i'},Y_{i'})\text{ on }A\land A\in R^1_i]$)
        \begin{align*}
            &\mathbb{E}_{(X_i,Y_i)}[q_{i',j,(X_i,Y_i)}\cdot p_{(X_i,Y_i)}] \\
            ={}&2^{-2n_im_i}\sum_{(X_i,Y_i)}\biggr(\Pr[B\mid A\text{ agrees with }(X^{\text{fixed}}_{\le i-1},Y^{\text{fixed}}_{\le i-1})\land A\text{ agrees with }(X_i,Y_i)]\biggr) \\
            ={}&2^{-2n_im_i}\sum_{(X_i,Y_i)}\frac{\Pr[B \land A\text{ agrees with }(X_i,Y_i)\mid A\text{ agrees with }(X^{\text{fixed}}_{\le i-1},Y^{\text{fixed}}_{\le i-1})]}{\Pr[A\text{ agrees with }(X_i,Y_i)\mid A\text{ agrees with }(X^{\text{fixed}}_{\le i-1},Y^{\text{fixed}}_{\le i-1})]} \\
            ={}&2^{-2n_im_i}\sum_{(X_i,Y_i)}\frac{\Pr[B\land A\text{ agrees with }(X_i,Y_i)\mid A\text{ agrees with }(X^{\text{fixed}}_{\le i-1},Y^{\text{fixed}}_{\le i-1})]}{2^{-2n_im_i}} \\
            ={}&\sum_{(X_i,Y_i)}\Pr[B\land A\text{ agrees with }(X_i,Y_i)\mid A\text{ agrees with }(X^{\text{fixed}}_{\le i-1},Y^{\text{fixed}}_{\le i-1})] \\
            ={}&\Pr[B\mid A\text{ agrees with }(X^{\text{fixed}}_{\le i-1},Y^{\text{fixed}}_{\le i-1})].
        \end{align*}
        \endgroup
        
        Using \eqref{equ:largeness_temp} and \eqref{equ:large_socratic_temp}, we can upper bound $\Pr[B\mid A\text{ agrees with }(X^{\text{fixed}}_{\le i-1},Y^{\text{fixed}}_{\le i-1})]$ by
        \begin{align}
            \label{equ:socratic_temp2}
            &\Pr[B\mid A\text{ agrees with }(X^{\text{fixed}}_{\le i-1},Y^{\text{fixed}}_{\le i-1})] \nonumber\\
            ={}&\Pr_A[Q_{i',j}\text{ solves }(X_{i'},Y_{i'})\text{ on }A\land A\in R^1_i\mid A\text{ agrees with }(X^{\text{fixed}}_{\le i-1},Y^{\text{fixed}}_{\le i-1})] \nonumber\\
            ={}&\Pr[Q_{i',j}\text{ solves }(X_{i'},Y_{i'})\text{ on }A\mid A\in R^1_i]\cdot \Pr[A\in R^1_i\mid A\text{ agrees with }(X^{\text{fixed}}_{\le i-1},Y^{\text{fixed}}_{\le i-1})] \nonumber\\
            \le{}&2^{-c_2m_i/n_i^{c_1}+O(n_{i-1}^{c_3})}\cdot 1.\tag{by \eqref{equ:large_socratic_temp}}
        \end{align}
        
        Now assume that the claim does not hold. Then we have that
        \begin{align*}
            &\mathbb{E}_{(X_i,Y_i)}[q_{i',j,(X_i,Y_i)}\cdot p_{(X_i,Y_i)}] \\
            &\ge \mathbb{E}_{(X_i,Y_i)}[q_{i',j,(X_i,Y_i)}\cdot 2^{-n_i^{c_3}}\mid p_{(X_i,Y_i)}\ge 2^{-n_i^{c_3}}]\cdot \Pr_{(X_i,Y_i)}[p_{(X_i,Y_i)}\ge 2^{-n_i^{c_3}}] \\
            &\ge \Pr_{(X_i,Y_i)}[q_{i',j,(X_i,Y_i)}> 2^{-c_2m_{i'}/n_{i'}^{c_1}+n_i^{c_4}}\mid p_{(X_i,Y_i)}\ge 2^{-n_i^{c_3}}]\cdot 2^{-n_i^{c_3}}\cdot 2^{-c_2m_{i'}/n_{i'}^{c_1}+n_i^{c_4}}\cdot p_{\text{large}} \\
            &\ge 2^{-c_2m_{i'}/n_{i'}^{c_1}+n_i^{c_4}\cdot(1-o(1))}.
        \end{align*}
        Since $c_4$ is sufficiently large, this contradicts \eqref{equ:socratic_temp2}.
    \end{claimproof}

    \paragraph{The union bound.} By applying an union bound over all $Q_{i',j}$, \autoref{clm:socratic_union} show that, when we randomly select an instance $(X_i,Y_i)$, conditioned on $p_{(X_i,Y_i)}\ge 2^{-n_i^{c_3}}$, with high probability, $q_{i',j,(X_i,Y_i)}$ is small for every $i_0\le i'\le i$ and $1\le j\le i'$. Combining this with \autoref{clm:large_fraction}, which shows that there exist instances that satisfy $p_{(X_i,Y_i)}\ge 2^{-n_i^{c_3}}$, we conclude the proof of \autoref{lem:extend_prefix}.
\end{proof}

\paragraph{Analysis for step $2$.} Finally, we show that the rectangle obtained from \autoref{lem:extend_prefix} satisfies Property \ref{property: refuting Socratic protocols}. Note that the property stated in \autoref{lem:extend_prefix} is only different from Property \ref{property: refuting Socratic protocols} on the requirement for protocols $Q_{i',j}$ with $h(i-1)<2^{i'}\le h(i)$. It then remains to show that:

\begin{claim}
    \label{clm:final_temp}
    Let $R_i$ be as constructed in \autoref{lem:extend_prefix}. For $i_0\le i'\le i$ and $1\le j\le i'$ such that $h(i-1)<2^{i'}\le h(i)$, and for any oracle $A\in R_i$, $Q_{i',j}$ does not solve $(X_{i'},Y_{i'})$ on $A$.
\end{claim}

\begin{proof}

    Note that, since $Q_{i',j}$ is simulating an algorithm $M$ that runs in time $2^{i'}$, it only uses the value of the oracle on the instances $(X_{\le i},Y_{\le i})$. Since all oracles in $R_i$ agree on $(X_{\le i},Y_{\le i})$, if $Q_{i',j}$ solves $(X_{i'},Y_{i'})$ on some oracle in $R_i$, then it does so on \emph{every} oracle in $R_i$. Therefore, we only have to show that
    \begin{align*}
        \Pr_A[Q_{i',j}\text{ solves }(X_{i'},Y_{i'})\text{ on }A\mid A\in R_i]\le 2^{-c_2m_{i'}/n_{i'}^{c_1}+n_i^{c_4}}<1.
    \end{align*}
    This holds by \autoref{clm:half_exp_math}.
\end{proof}

%% file: BFT_is_robust.tex
\section{On the Circuit Lower Bound of Buhrman--Fortnow--Thierauf}
\label{app:bft}

We verify that the proof of~\cite{BuhrmanFT98} indeed gives a sub-half-exponential circuit lower bound for $\E \cup \RobMAE$ that algebrizes. We also interpret this proof as a ``win-win'' algorithm for $\MissingString$ in the algebraic decision tree model of~\cite{DBLP:journals/toct/AaronsonW09}.

\subsection{Some Preliminaries}
In this section, when we relativize space-bounded computations, we take the query tape into account when measuring space complexity. That is, a $\SPACE^A[s(n)]$ machine can only make queries of length at most $s(n)$ to its oracle $A$. The class $\PSPACE^A$ in our context equals to the class $\PSPACE^{A[\poly]}$ in~\cite{DBLP:journals/toct/AaronsonW09}.

We say that $h(n)$ is a \emph{sub-half-exponential} function, if $h(h(n)^c) \le 2^n$ for every constant $c\ge 1$. We say that $h(n)$ is \emph{nice}, if there is a deterministic algorithm that on input $n$, outputs the value of $h(n)$ in $\poly(n)$ time.

We use the notation $\langle P, V\rangle(x)$ to denote the outcome of the interactive protocol with prover $P$, verifier $V$, and common input $x$ (which is a random variable). We use $\IP^A$ to denote the class of languages that can be decided by an interactive protocol where the prover is computationally unbounded and the verifier is a randomized polynomial-time machine with oracle access to $A$. Formally, $L\in\IP^A$ if there exists a randomized polynomial-time verifier $V$ with oracle access to $A$ such that for every $x\in \{0, 1\}^*$:
\begin{itemize}
    \item If $x \in L$, then there is an (unbounded) prover $P$ such that $\Pr[\langle P, V\rangle(x)\text{ accepts}] = 1$.
    \item If $x\not\in L$, then for every (unbounded) prover $P$, $\Pr[\langle P, V\rangle(x)\text{ accepts}] \le 1/3$.
\end{itemize}

We need the algebrizing version of $\IP = \PSPACE$~\cite{LundFKN92, Shamir92}.
\begin{theorem}[{\cite{DBLP:journals/toct/AaronsonW09}}]\label{thm: IP = PSPACE algebrizes}
	For every oracle $A$ and every multilinear extension $\widetilde{A}$ of $A$, 
	\[\PSPACE^A \subseteq \IP^{\widetilde{A}}.\]
    Moreover, the honest $\IP$ prover can be implemented in $\PSPACE^{\widetilde A}$.
\end{theorem}

We remark that \autoref{thm: BFT} can be proved by either $\IP = \PSPACE$, $\MIP = \NEXP$~\cite{DBLP:journals/cc/BabaiFL91}, or the PCP theorem~\cite{DBLP:journals/jacm/AroraS98, DBLP:journals/jacm/AroraLMSS98, DBLP:journals/jacm/Dinur07}. In fact, in \autoref{sec: BFT as algebraic query algorithms}, we will see another proof using PCP of \emph{Proximity}~\cite{DBLP:journals/siamcomp/Ben-SassonGHSV06}.

\subsection{Reviewing the BFT Proof}\label{sec: BFT is robust}

\begin{theorem}\label{thm: BFT}
	Let $h(n)$ be a nice and sub-half-exponential function. For every oracle $A$ and every multilinear extension $\widetilde{A}$ of $A$,
	\[(\E^{\widetilde{A}} \cup \RobMAE^{\widetilde{A}}) \not\subseteq \SIZE^A[h(n)].\]
\end{theorem}
\begin{proof}
	Assuming $\E^{\widetilde{A}} \subseteq \SIZE^A[h(n)]$, we will prove that $\RobMAE^{\widetilde{A}} \not\subseteq \SIZE^A[h(n)]$.
	
	\def\Lhard{L_{\rm hard}}
	\def\Mhard{M_{\rm hard}}
	First, let $\Lhard^A$ denote the language whose truth table on input length $n$ is the lexicographically smallest length-$2^n$ truth table that requires $A$-oracle circuit complexity more than $h(n)$. Clearly, $\Lhard \in \SPACE^A[\tilde{O}(h(n))] \setminus \SIZE^A[h(n)]$. (We stress two details here. First, there is a fixed $\SPACE[\tilde{O}(h(n))]$ oracle machine $\Mhard^{(-)}$ independent of the oracle $A$ such that $\Mhard^A$ computes $\Lhard^A$. Second, $\Mhard^A$ only makes queries of length at most $h(n)$ to the oracle $A$.) Let $(\Lhard')^A := \{(x, b): \Lhard^A(x) = b\}$, then it follows from \autoref{thm: IP = PSPACE algebrizes} and a padding argument that $(\Lhard')^A$ admits an interactive proof $\langle P, V\rangle$ where $V$ runs in randomized $\poly(h(n))$ time and the honest $P$ runs in $2^{\poly(h(n))}$ time; both the verifier and the prover need access to $\widetilde A$. %
	
	Now we consider the following $\RobMAE^{\widetilde A}$ algorithm $M$. The algorithm receives an input $x\in\{0, 1\}^n$ and a proof $\pi$ from Merlin, where $\pi$ consists of a bit $b$, as the purported value of $\Lhard(x)$, and an $A$-oracle circuit $P'$ of $2^n$ size, treated as an $\IP$ prover for $\Lhard$. (Note that the input length of $P'$ is $\poly(h(n))$.) Then, $M$ accepts if and only if $\langle P', V\rangle(x, b)$ accepts, i.e., $P'$ convinces $V$ that $\Lhard(x) = b$. Note that $M$ corresponds to the pair of verifiers $(V_0, V_1)$ where $V_b$ accepts iff $M$ accepts and the first bit of $\pi$ is equal to $b$.
	
	We can see that:\begin{itemize}
		\item First, $M$ is an $(\MA_\E\cap\coMA_\E)^{\widetilde A}$ algorithm that is robust w.r.t.~$h(n)$ on every input length $n$. To see the robustness of $M$, fix an arbitrary oracle $B$ and its multilinear extension $\widetilde B$, and consider a truth table $w$ that has $h(n)$-size $B$-oracle circuits. Then, there exists an input $x \in \{0, 1\}^n$ such that $\Lhard^B(x) \ne w_x$. By the soundness of $V$, it holds that for every circuit $P'$ that Merlin could send,
        \[\Pr[\langle P', V\rangle(x, w_x)\text{ accepts}] \le 1/3,\]
        Hence $V^{\widetilde B}_{w_x}$ rejects $x$.

		\item Second, if $\E^{\widetilde A} \subseteq \SIZE^A[h(n)]$, then $M^{\widetilde A}$ computes a function whose circuit complexity is larger than $h(n)$. This is because there is a prover $P$ running in $2^{\poly(h(n))}$ time with oracle access to $\widetilde A$ such that for every input $x\in\{0, 1\}^n$, $\Pr[\langle P(x, -), V\rangle(x, b)\text{ accepts}] \ge 2/3$. Therefore, under our assumption that $\E^{\widetilde A} \subseteq \SIZE^A[h(n)]$ and noting that the input length of $P$ is $\poly(h(n))$, $P$ can be implemented by an $A$-oracle circuit of size $h(\poly(h(n))) \le 2^n$ and Merlin can send this circuit. It follows that every $x\in\{0, 1\}^n$ is easy and $M^{\widetilde A}$ computes $\Lhard^A$.\qedhere
	\end{itemize}
\end{proof}

\subsection{BFT as Algorithms for Missing-String}\label{sec: BFT as algebraic query algorithms}

In this sub-section, we interpret the lower bound in~\cite{BuhrmanFT98} as algebraic query algorithms for $\MissingString$, more clearly illustrating its win-win analysis.%

\paragraph{PCP of proximity.} Our algebraic query algorithm makes use of \emph{PCPPs} (probabilistically checkable proofs of proximity)~\cite{DBLP:journals/siamcomp/Ben-SassonGHSV06}. Let $L \subseteq \{0, 1\}^* \times \{0, 1\}^*$ be a \emph{pair language} where the first input is called the \emph{explicit} input and the second input is called the \emph{implicit} input. A PCPP verifier for $L$ is a query algorithm $V$ with the following specification:
\begin{itemize}
    \item {\bf Inputs:} $V$ takes $(x, r)$ as inputs where $x$ is the explicit input for $L$ and $r$ is the internal randomness of $V$.
    \item {\bf Oracles:} $V$ has oracle access to the implicit input $y$ and a PCPP proof $\pi$.
    \item {\bf Completeness:} If $(x, y) \in L$, then there exists an honest PCPP proof $\pi$ such that
    \[\Pr_r[V^{y, \pi}(x, r)\text{ accepts}] = 1.\]
    \item {\bf Soundness:} If $y$ is $\delta$-far from the set $\{y': (x, y') \in L\}$, then for every PCPP proof $\pi$,
\[\Pr_r[V^{y, \pi}(x, r)\text{ accepts}] \le 1/3.\]
    Here, $\delta > 0$ is called the \defn{proximity parameter} for the PCPP.
\end{itemize}

We need the following construction of PCPP in~\cite{DBLP:journals/amai/Mie09}:
\begin{theorem}\label{thm: Mie PCPP}
    Let $L$ be a pair language with explicit input length $n$ and implicit input length $K$, such that $L\in\NTIME[T(n+K)]$ for some non-decreasing function $T(\cdot)$. For every constant $\delta > 0$, $L$ admits a PCPP verifier with proximity parameter $\delta$, randomness complexity $|r| = \log T(n) + O(\log\log T(n))$, query complexity $O(1)$, and verification time $\poly(n, \log K, \log T(n+K))$.
    
    Moreover, for every $(x, y) \in L$, given a valid witness $w$ for $(x, y)$, a valid PCPP proof $\pi$ for $(x, y)$ can be computed in deterministic polynomial time.
\end{theorem}

\def\Enc{\mathsf{Enc}}

\paragraph{Algebraic query complexity.} In the \emph{algebraic query model}, the algorithm has query access to the multilinear extension of its input. Thanks to the PCPP, the only property of multilinear extensions we need in this sub-section is that it is a \emph{systematic error-correcting code}.

An error-correcting code with \defn{distance $\delta > 0$} is a function $\Enc: \{0, 1\}^n \to \{0, 1\}^\ell$ such that for every two different $n$-bit inputs $x_1, x_2$, the number of indices $i\in [\ell]$ such that $\Enc(x_1)_i \ne \Enc(x_2)_i$ is at least $\delta\cdot \ell$. The code is \defn{systematic} if the first $n$ bits of $\Enc(x)$ are always equal to $x$ itself.

Fix any systematic error-correcting code $\Enc$ (such as the one that maps an input to its multilinear extension), an algebraic query algorithm for a problem $L$ is an algorithm that receives $\Enc(x)$ as inputs and produces the output $L(x)$. Note that since $\Enc$ is systematic, the algorithm can access its input $x$ in verbatim by looking at the first $n$ bits of $\Enc(x)$.

\def\calX{\mathcal{X}}

Now we show that~\cite{BuhrmanFT98} can be interpreted as win-win algorithms for $\MissingString$ in the algebraic query model:
\begin{theorem}\label{thm: BFT as algebraic query algorithms}
	Let $h(n)$ be a nice and sub-half-exponential function, $h'(n) = h(n)^C$ for some large enough universal constant $C$. Let $\calX_1$ be an instance of $\MissingString(n, h(n))$, and $\calX_2$ be an instance of $\MissingString(h'(n), 2^n)$. There are two (query) algorithms $\calA_1$ and $\calA_2$ which both receives $\Enc(\calX_1)$ and $\Enc(\calX_2)$ as inputs, such that:
	\begin{itemize}
		\item $\calA_1$ is a robust $\MA^\dt$ algorithm with query complexity $\poly(n)$ that outputs a string of length $n$: No malicious prover (Merlin) can convince the algorithm to output a wrong answer for $\calX_1$ except with probability $\le 1/3$.
		\item $\calA_2$ is a deterministic algorithm with query complexity $\poly(h(n))$ that outputs a string of length $h'(n)$.
		\item For every inputs $\calX_1, \calX_2$, either $\calA_1(\calX_1, \calX_2)$ solves the instance $\calX_1$, or $\calA_2(\calX_1, \calX_2)$ solves the instance $\calX_2$.
	\end{itemize}
\end{theorem}
\begin{proof}
	Let $L$ be the language consisting of $(a, \Enc(\calX))$, where $a\in\{0, 1\}^n$ is the explicit input, $\calX$ is a $\MissingString(n, h(n))$ instance and $\Enc(\calX)$ is the implicit input, and $(a, \Enc(\calX))\in L$ if and only if $a$ is the lexicographically smallest missing string of $\calX$. Let $\delta$ be the distance of $\Enc$, and $V$ be the PCPP verifier for $L$ with robustness parameter $\delta$ as specified in \autoref{thm: Mie PCPP}.

	Let $a^*$ be the lexicographically smallest missing string of $\calX_1$, and $\pi$ be the PCPP proof for $(a^*, \Enc(\calX_1))\in L$. If $C$ is a large enough constant, then $|\pi| \le h'(n)$. The algorithm $\calA_2$ simply outputs $\pi$ (padded with zeros if $|\pi| < h'(n)$), which can be computed in deterministic $\poly(h(n))$ time. Clearly, if $\pi$ does not appear in $\calX_2$, then $\calA_2(\calX_1, \calX_2)$ solves the instance $\calX_2$. %

	The $\MA^\dt$ algorithm $\calA_1$ receives a proof from Merlin, which contains a string $a\in\{0, 1\}^n$ and an index $j \in [2^n]$. The intended meaning is that $a$ is the lexicographically smallest missing string of $\calX_1$, and that $\calX_2[j]$ (the $j$-th string in $\calX_2$) is a valid PCPP proof that $(a, \Enc(\calX_1)) \in L$. The $\MA^\dt$ algorithm runs the PCPP verifier $V^{\Enc(\calX_1), \calX_2[j]}(a)$ and accepts if and only if the PCPP verifier accepts. Clearly, if $\pi$ appears in $\calX_2$, then $\calA_1(\calX_1, \calX_2)$ solves the instance $\calX_1$. Moreover, the soundness of the PCPP verifier implies that no malicious Merlin can convince $\calA_1$ to output a wrong answer except with a small probability: if $a$ is not the lexicographically smallest missing string of $\calX_1$, then $\Enc(\calX_1)$ is $\delta$-far from $\{\Enc(\calX): (a, \Enc(\calX)) \in L\}$, hence the verifier rejects with high probability.
\end{proof}